\documentclass[%
 reprint,
 amsmath,amssymb,
 aps,
 nofootinbib,
 onecolumn,
longbibliography,
preprintnumbers,
floatfix
]{revtex4-2}

\usepackage{graphicx}
\usepackage{dcolumn}
\usepackage{bm}
\usepackage{bbm}
\usepackage{xspace}
\usepackage{verbatim}
\usepackage{subcaption}
\DeclareCaptionJustification{justified}{\leftskip=0pt\rightskip=0pt\parfillskip=0ptplus1fil\relax}
\usepackage[normalem]{ulem}
\usepackage{xcolor}
\newcommand{\Sym}[1]{\mathcal{S}(#1)}

\newcommand{\vecSop}{\hat{\vec{\sigma}}}
\newcommand{\vecJop}{\hat{\vec{S}}}

\newcommand{\rotdim}{d_{r}}
\newcommand{\exch}[1]{\hat{\rho}(#1)}

\newcommand{\derange}[1]{d(#1)}

\usepackage[
colorlinks=true,
linkcolor=blue,
citecolor=blue,
filecolor=blue,
urlcolor=blue,
]{hyperref}
\usepackage[capitalise]{cleveref}
\usepackage{pgfplots}
\newcommand{\SU}{\operatorname{SU}}
\renewcommand{\O}{\operatorname{O}}
\newcommand{\SO}{\operatorname{SO}}
\newcommand{\tr}{\operatorname{tr}}
\renewcommand{\d}{\mathrm{d}}
\renewcommand{\i}{\mathrm{i}}
\newcommand{\e}{\mathrm{e}}

\usepackage{amsthm}
\theoremstyle{plain}
\newtheorem{lemma}{Lemma}

\newcommand{\UNITN}{{Dipartimento di Fisica, University of Trento, via Sommarive 14, I–38123, Povo, Trento, Italy}}
\newcommand{\TIFPA}{INFN-TIFPA Trento Institute of Fundamental Physics and Applications,  Trento, Italy}
\newcommand{\LANL}{Theoretical Division, Los Alamos National Laboratory, Los Alamos, New Mexico, 87545}
\begin{document}

\preprint{LA-UR-24-25341}

\title{Scattering Neutrinos, Spin Models, and Permutations}

\author{Duff Neill}%
 \email{dneill@lanl.gov}
\author{Hanqing Liu}%
 \email{hanqing.liu@lanl.gov}
\author{Joshua Martin}%
\affiliation{%
\LANL
}%
\author{Alessandro Roggero}%
  \affiliation{%
\UNITN\\
\TIFPA
}%
\date{\today}

\begin{abstract}
We consider a class of Heisenberg all-to-all coupled spin models inspired by neutrino interactions in a supernova with $N$ degrees of freedom. These models are characterized by a coupling matrix that is relatively simple in the sense that there are only a few, relative to $N$, non-trivial eigenvalues, in distinction to the classic Heisenberg spin-glass models, leading to distinct behavior in both the high-temperature and low-temperature regimes. When the momenta of the neutrinos are uniform and random in directions, we can calculate the large-$N$ partition function for the quantum Heisenberg model. In particular, the high-temperature partition function predicts a non-Gaussian density of states, providing interesting counter-examples showing the limits of general theorems on the density of states for quantum spin models. We can repeat the same argument for classical Heisenberg models, also known as rotor models, and we find the high-temperature expansion is completely controlled by the eigenvalues of the coupling matrix, and again predicts non-Gaussian behavior for the density of states as long as the number of eigenvalues does not scale linearly with $N$. Indeed, we derive the amusing fact that these \emph{thermodynamic} partition functions are essentially the generating function for counting permutations in the high-temperature regime. Finally, for the case relevant to neutrinos in a supernova, we identify the low-temperature phase as a unique state with the direction of the momenta of the neutrino dictating its coherent state in flavor-space, a state we dub the ``flavor-momentum-locked'' state.    
\end{abstract}
\maketitle


\section{Introduction}
The interaction between neutrinos in a supernova, neglecting non-forward scattering, is often modeled by a Heisenberg magnet with all-to-all couplings, see Refs.~\cite{Pehlivan:2011hp,PhysRevD.99.123013}. In the ultra-relativistic limit, the masses of the neutrinos can be neglected, and the scattering Hamiltonian (the 4-fermi interaction from integrating out the Z-boson) acquires a $\SU(n_f)$ global symmetry, $n_f$ being the number of neutrino flavors, allowing one to think about the system of scattering neutrinos as a $\SU(n_f)$ Heisenberg magnet, as long as we ignore scattering into new momentum modes. For a picture of what the full Hamiltonian treatment should look like, see Ref.~\cite{Cirigliano:2024pnm}. Due to the inevitable presence of both coherent and incoherent processes included in the calculation, this forward scattering model can be criticized as a model for neutrino physics (since it neglects incoherent effects coming from non-forward scattering), Refs.~\cite{PhysRevD.107.123004,johns2023neutrino}. However, we are content here, for it gives us a new class of spin models: in particular, it will give us a new approach to mean-field models of condensed matter physics, Refs.~\cite{1978JPhC...11L.871D,1980JPhC...13L.655B,PhysRevLett.70.3339,1993PhRvL..70.4011Y,Georges:2000zfr}, where the control parameter for the phase-diagram of the spin-model is the number of non-zero eigenvalues in the coupling matrix. 

In Ref.~\cite{Martin:2023gbo}, motivated by the eigenstate thermalization hypothesis of Refs.~\cite{PhysRevA.43.2046,Srednicki1994,Srednicki_1996}, it was suggested late-time values of local observables for neutrinos in a supernova should be given by the corresponding thermal averages. This motivated our study of the phase diagrams for such systems, but the commonly used models in the condensed matter literature fail for our purposes: their couplings are taken to be generic in a mean-field sense, that is, Gaussian distributed about zero with a variance of $1/\sqrt{N}$, $N$ the number of spins. This coupling matrix leads to a Gaussian partition function at high temperatures, and a spin-glass phase transition at low temperatures. In contrast, the neutrino models have what we could call ``asymptotically simple'' coupling matrices, that is, coupling matrices that do not have a number of non-zero eigenvalues scaling proportional to the number of spins. In general, these models will have non-Gaussian density of states, quantum or classical, and do not appear to have any spin-glass phase, but either have a low-temperature phase which we dub the ``flavor-momentum locked'' (FML) phase~\footnote{The name is inspired from the neutrino model: the coupling matrix is derived from the 4-momenta of the neutrinos, and the flavor state of the neutrino is locked to it.}, or exhibit no phase transition as we approach low temperatures if the coupling matrix is positive semi-definite.

In this paper, we work out the large-$N$ limit of the high-temperature expansion of the partition function for the quantum Heisenberg magnet for a set of couplings naturally suggested by the neutrino problem, where $N$ is the total number of neutrinos/sites of the magnet. We then confirm this expansion by diagonalizing the Hamiltonian over many realizations of the couplings, and verifying the predicted density of states. The partition function thus derived has a pole at a finite temperature, leading to a phase transition, and we identify the low-temperature phase. Further, this high-temperature expansion is easily generalized to a classical Heisenberg magnet with a much more generic coupling matrix, and we show the partition function at high temperatures is completely determined by the eigenvalues of the coupling matrix. Through the Leib inequalities of Ref.~\cite{1973CMaPh..31..327L}, we can give an argument that we should not expect a Gaussian density of states for any of the corresponding quantum models. We then determine the zero-temperature phase diagram for the classical neutrino Hamiltonian with slightly more general couplings, and show that when the coupling is exactly the neutrino coupling, the ground state is in the FML phase. We also argue how this zero-temperature phase diagram is consistent with the high-temperature expansion. The success of the high-temperature expansion seems to be tied to our assumption that the number of non-trivial eigenvalues of the coupling matrix is much less than the number of spins, and numerical simulations suggest the classical spin-glass mean-field analysis (Ref.~\cite{1978JPhC...11L.871D}) holds when the number of eigenvalues is proportional to the system size.   

The paper is organized as follows: in \cref{sec:systems} we define the set of systems we wish to investigate, both classical and quantum. In \cref{sec:rotor_coupling},  we work out the high-temperature expansion of the partition function for a particular quantum model. In \cref{sec:classical}, the method for this expansion is then generalized to a larger class of classical models when the number of non-zero eigenvalues in the coupling matrix is small compared with $N$. Generically, our results for the large-$N$ limit of these classical Heisenberg spin models predict a pole in the partition function determined by the largest eigenvalues of the coupling matrix, signaling a phase transition. With applications to neutrino physics in mind, we return to the neutrino Hamiltonian in \cref{sec:FML}, and we work out the structure of the low-temperature phase. While we pursue the example of where the momenta are uniformly distributed on the sphere, we have every expectation that both the phase-transition and the FML phase will be obtained for any distribution of the momenta. 

\section{The Systems}
\label{sec:systems}
In what follows, we let $\Lambda=\{1,\cdots,N\}$ be the set of all site labels, and we define the sums:
\begin{align}
\sum_{\{i_1,\cdots,i_L\}\subset \Lambda}\equiv\sum_{1\leq i_1<\cdots<i_L\leq N}\,.
\end{align}
That is, we are summing over all subsets of size $L$, with no repeated labels in the subset. To make the notation more compact, we will often drop the $\subset \Lambda$, if no confusion can result.

Inspired by the modeling of forward scattering neutrino interactions in a supernova as a quantum many-body spin system, involving all-to-all interactions between pairs of qubits (using a two-flavor approximation), we are ultimately concerned with the Hamiltonian of the form:
\begin{align}\label{eq:neutrino_H}
\hat{H}^{\nu\nu}= \sum_{\{i,j\}\subset \Lambda}\frac{1}{2N}\big(\mu_1 + \mu_2 \vec{v}_i\cdot\vec{v}_j\big) \vecSop_{i}\cdot\vecSop_j\,.
\end{align}
The $\vecSop_{i}$ are vectors of the Pauli matrices acting on site $i$, and $\vec{v}_i$ are unit vectors on the 2-dimensional sphere, distributed according to some dynamical process in the supernova, corresponding to the spatial direction of the approximately null-momentum of the neutrino. The quantum state at site $i$ is supposed to represent the flavor of the neutrino occupying that momentum mode. In what follows, we will simply refer to the $\vec{v}_i$ as the momentum. Typically, one takes $\mu_1=-\mu_2=\mu$, which results from the underlying Lorentz invariance of the neutrino interactions. Then $\mu$ encodes the densities of the neutrinos that are interacting. However, we keep it a separate parameter, as the $\mu_1$ can be related to a chemical potential for the total angular momentum of the thermal state, as used in Ref.~\cite{Martin:2023gbo}. 

Our goal is to understand the phase diagram of the partition function:
\begin{align}\label{eq:neutrino_Z}
    Z_{q}^{\nu\nu}\big(\beta;\mu_1,\mu_2,\{\vec{v}_i\}\big)&=2^{-N}\text{tr}[\exp(-\beta\hat{H}^{\nu\nu})]\,,
\end{align}
as a function of temperature and the coupling strengths $\mu_1$ and $\mu_2$, for a given distribution for the momenta. To accomplish this, we will generalize the model and consider related Hamiltonian systems. First, we generalize $\vec{v}_i$ be a random unit vector embedded in a $d$-dimensional space, corresponding to a point uniformly distributed on the $(d-1)$-dimensional sphere $\mathbb{S}^{d-1}$, i.e., the unit sphere embedded in $d$-spatial dimensions. This will allow the coupling matrix $\vec{v}_i\cdot\vec{v}_j$ to be any positive semi-definite matrix, if we take to be any integer $d\leq N$, up to shifts along the diagonals (the dropped terms where $i=j$ in Eq.~\eqref{eq:neutrino_H}) that do not effect the underlying physics. We will have an interesting system where the quantum degrees of freedom are effectively coupled to a set of classical degrees of freedom corresponding to a classical rotor model, see Ref.~\cite{1978JPhC...11L.871D}. In this case, we will develop an analytic expression for the high-temperature phase of the model, exact in the large-$N$ limit.

Alternatively, we can keep the couplings generic, and demote the quantum operators $\vecSop_{i}$ to classical unit vectors $\vec{S}_i$, and then use Ref.~\cite{1973CMaPh..31..327L} to give bounds on the quantum partition function of Eq.~\eqref{eq:neutrino_Z}, in terms of the partition function for the classical rotor model. This leads us to consider the classical Hamiltonian and its resulting partition function:
\begin{align}\label{eq:_dyn_cl_rotor_H}
H_{\rm cl.}&=\sum_{\{i,j\}\subset \Lambda}  J_{ij}\vec{S}_i\cdot \vec{S}_j\,,\\
\label{eq:_dyn_cl_rotor_Z}Z_{\rm cl.}(\beta;J)&=\int \prod_{i=1}^{N}\frac{\d^{\rotdim}\vec{S}_i}{\Omega_{d_r}}\delta(1-\vec{S}_i^2)\exp\big(-\beta H_{\rm cl.}\big)\,,
\end{align}
assuming $\Omega_{\rotdim}$ is the area of the sphere, so $Z_{\rm cl.}(0,J)=1$, and we have generalized our degrees of freedom to unit vectors embedded in $\rotdim$-dimensions, the vectors $\vec{S}_i$. In particular, if $J_{ij}=\frac{1}{2N}(\mu_1 + \mu_2 \vec{v}_i\cdot\vec{v}_j)$, and we set $\rotdim=3$, we can define the classical neutrino partition function:
\begin{align}
    H_{\rm cl.}^{\nu\nu}&=\sum_{\{i,j\}\subset \Lambda}  \frac{1}{2N}\big(\mu_1  + \mu_2 \vec{v}_i\cdot\vec{v}_j\big)\vec{S}_i\cdot \vec{S}_j\,, \label{eq:Hnu-cl}\\
    Z_{\rm cl.}^{\nu\nu}\big(\beta;\mu_1,\mu_2,\{\vec{v}_i\}\big)&=\int \prod_{i=1}^{N}\frac{\d^{\rotdim}\vec{S}_i}{\Omega_{d_r}}\delta(1-\vec{S}_i^2)\exp\big(-\beta H_{\rm cl.}^{\nu\nu}\big)\,.
\end{align}

Using the inequalities in Ref.~\cite{1973CMaPh..31..327L}, we can bound the quantum partition function \cref{eq:neutrino_Z} as
\begin{align}\label{eq:Leib_ineq}
    Z_{\rm cl.}^{\nu\nu}\big(\beta;\mu_1,\mu_2,\{\vec{v}_i\}\big)\leq Z_{q}^{\nu\nu}\big(\beta;\mu_1,\mu_2,\{\vec{v}_i\}\big)\leq Z_{\rm cl.}^{\nu\nu}\big(9\beta;\mu_1,\mu_2,\{\vec{v}_i\}\big)\,.
\end{align}
Or more generally, if we promote all sites in the quantum partition function Eq.~\eqref{eq:neutrino_Z} to a spin-$s$ representation with generators $\vecJop_{i}$:
\begin{align}\label{eq:neutrino_H_spin-s}
\hat{H}^{\nu\nu}_s&= \sum_{\{i_1,i_2\}\subset \Lambda}\frac{1}{2N}\big(\mu_1 + \mu_2 \vec{v}_i\cdot\vec{v}_j\big) \vecJop_{i}\cdot\vecJop_j\,,\\ 
[\hat{S}_{i}^{a},\hat{S}_{j}^{b}]&=i\delta_{ij}\sum_{c}\varepsilon^{abc}\hat{S}^c_{i}\,.
\end{align}
Then we have the following inequalities:
\begin{align}\label{eq:leib}
    Z_{\rm cl.}^{\nu\nu}\big(s^2\beta;\mu_1,\mu_2,\{\vec{v}_i\}\big)\leq Z_{q}^{\nu\nu}\big(\beta;\mu_1,\mu_2,\{\vec{v}_i\};s\big)\leq Z_{\rm cl.}^{\nu\nu}\big((s+1)^2\beta;\mu_1,\mu_2,\{\vec{v}_i\}\big)\,.
\end{align}
We have explicitly indicated the representation of the sites in the arguments of the quantum partition function defined with Eq. \eqref{eq:neutrino_H_spin-s}.\footnote{When comparing with the neutrino Hamiltonian using Pauli-matrices initially, rather than the angular momentum operators, one should be mindful of factors of 2 when using these inequalities.}

Finally, for compact notation, we define the partition functions:
\begin{align}\label{eq:rotor_q_z}
    Z^{\rm rot.}_{q}\big(\beta;\mu\big)=Z_{q}^{\nu\nu}\big(\beta;0,\mu,\{\vec{v}_i\}\big)\,,\\
 \label{eq:rotor_cl_z}   Z^{\rm rot.}_{\rm cl.}\big(\beta;\mu\big)= Z_{\rm cl.}^{\nu\nu}\big(\beta;0,\mu,\{\vec{v}_i\}\big)\,.
\end{align}
These partition functions are particularly interesting, since the coupling matrix is derived from a \emph{fixed} state in a classical rotor model. This automatically implies that the coupling matrices for these Hamiltonians are positive semi-definite, up to the diagonal terms. 

In what follows, we will develop the high-temperature expansion for \cref{eq:rotor_q_z} as $N\rightarrow\infty$, assuming that the momenta are randomly distributed uniformly on the sphere. This derivation will also work for the high-temperature expansion for $Z_{\rm cl.}(\beta;J)$ with a generic coupling matrix, where the roles of the coupling matrix and the quantum operators get swapped.  In the case that $\mu_2 < 0$ and $3\mu_1 > \mu_2$, and $d=3$, for the thermal system defined by Eq.~\eqref{eq:neutrino_Z} (thus including Eq.~\eqref{eq:rotor_q_z}), we will argue for a phase transition to what we call the flavor-momentum locked (FML) state occurs. This FML state can be understood semi-classically from the partition function in Eq.~\eqref{eq:rotor_cl_z} when $d=\rotdim$, where we can easily perform numerical simulations. Using Leib's inequalities of Eq.~\eqref{eq:leib} we can argue for the same phase-transition in the quantum case. 

In this paper, we take $g(x)\sim O(f(x)), x\in R$ to mean there is a constant $c$ such $|g(x)|\leq c f(x)$ in the region $R$.

\section{High-Temperature Expansion for Rotor Distributed Couplings}\label{sec:rotor_coupling}

First, we want to find the high-temperature expansion for the partition function of Eq.~\eqref{eq:rotor_q_z}. The type of Hamiltonians we then want to look at is:
\begin{align}
  \hat{H}^{\rm rot.} &= \sum_{\{i,j\}} \frac{\mu}{2N}\vec{v}_{i}\cdot\vec{v}_{j}\hat{O}_{ij}\,,
\end{align}  
where $\hat{O}_{ij}$ is some two-body operator acting on sites $i,j$ of the system. In particular, we will be interested in the cases:\footnote{We note that a very similar Hamiltonian to Eq. \eqref{eq:SK_H} was written down in Refs.\cite{Wei:2019rqy} and \cite{Delgado:2022snu} to determine the thrust-axis of a high energy scattering event in quantum chromodynamics (QCD) via quantum annealing, though with an underlying distribution of momenta set by the scattering processes of QCD.}
\begin{align}\label{eq:ex_H}
    \hat{H}^{\rm rot.}_{q} &= \sum_{\{i,j\}}\frac{\mu}{N}\vec{v}_{i}\cdot\vec{v}_{j}\exch{ij}\,,\\
    \hat{H}^{\rm rot.}_{\rm SK} &= \sum_{\{i,j\}}\frac{\mu}{2N}\vec{v}_{i}\cdot\vec{v}_{j}\hat{\sigma}_{i}^{z}\hat{\sigma}_{j}^{z}\,.\label{eq:SK_H}
\end{align}  
$\exch{ij}$ is the swap operator acting on the sites $i,j$, that interchanges their states, whose expression in terms of normal Pauli matrices is:
\begin{align}\label{eq:pauli_complete}
\hat{\vec{\sigma}}_{i}\cdot\hat{\vec{\sigma}}_j+\mathbbm{1}&=2\exch{ij}\,,\\
\exch{ij}|a\rangle_{i}\otimes |b\rangle_{j}&=|b\rangle_{i}\otimes |a\rangle_{j}\,.
\end{align}
The coupling matrix is now distributed according to the rotor model: the vectors $\vec{v}_{i}$ are of unit length and uniformly distributed on the $(d-1)$-sphere. For the forward scattering neutrino model, we often assume two neutrino flavors, but we give the generalization to $n_f$ flavors as appropriate. Our Hilbert space is then the tensor product of $N$ qudits with local dimension $n_f$, giving a total dimension of $n_f^N$ to the Hilbert space.

We now work out the high-temperature expansion of the quantum partition function when the couplings are vectors on the sphere $\mathbb{S}^{d-1}$. We will need elementary facts concerning the conjugacy class structure of symmetric groups summarized in \cref{sec:permutation_review}, as this will be directly relevant for the high-temperature expansion. Textbook treatments can be found in Refs.~\cite{grouptheory} or \cite{grouptheoryII}.

\subsection{The Annealed High-Temperature Partition Function}\label{sec:annealed}
At high temperatures, as $N\rightarrow \infty$, we can use the central limit theorem to select out a subset of diagrams that dominate the expansion of the exponential as $\beta\rightarrow 0$. Further, we expect for a given realization of the couplings, that the calculated partition function in the region of validity of the high-temperature expansion will be the same as the partition function averaged over the couplings:
\begin{align}
\lim_{N\rightarrow\infty}Z(\beta,\{\vec{v}_i\})\approx \Big\langle\!\!\Big\langle Z(\beta,\{\vec{v}_i\}) \Big\rangle\!\!\Big\rangle=\int \prod_{i=1}^{N}\d^{d}\vec{v}_i p(\vec{v}_i) Z(\beta,\{\vec{v}_i\})\,.
\end{align}
The double angle brackets indicate we average over the momentum vectors $\vec{v}_i$, and $p(\vec{v}_i)$ is the relevant distribution for the momenta, which for us will be $p(\vec{v}_i)=\delta(1-\vec{v}_i^2)/\Omega_d$. This is the so-called \emph{annealed} partition function, and this approximation can be expected to hold when we are allowed to expand the exponential of the partition function as a power series in the inverse temperature $\beta$. 
This can be expected to fail when we are outside the realm of validity of the high-temperature expansion, since this averaging is equivalent to invoking the central limit theorem order-by-order in the high-temperature expansion. That is, we expand in $\beta$, and use the central limit theorem to evaluate the sums over the couplings, thus taking the large $N$ limit after we have expanded. For a discussion of quenched vs. annealed partition functions, we refer to Ref.~\cite{Baldwin:2019dki}, as well as cautionary tales about the summation of the resulting series after the large-$N$ limit. 

We can show how the partition function simplifies in the large-$N$ limit at orders $\beta$ and $\beta^2$:
\begin{align}
    Z(\beta) &= 1-\beta\sum_{\{i,j\}}\Big(\frac{\mu}{2N}\vec{v}_{i}\cdot\vec{v}_{j}\Big)\frac{\text{tr}[\hat{O}_{i j}]}{\text{tr}[\mathbbm{1}]}+\frac{\beta^2}{2}\sum_{\{i_1,j_1\}}\sum_{\{i_2,j_2\}}\Big(\frac{\mu^2}{4N^2}\vec{v}_{i_{1}}\cdot\vec{v}_{j_{1}}\vec{v}_{i_{2}}\cdot\vec{v}_{j_{2}}\Big)\frac{\text{tr}[\hat{O}_{i_{1} j_{1}}\hat{O}_{i_{2} j_{2}}]}{\text{tr}[\mathbbm{1}]}+\cdots
\end{align}
We recall the notation $\{i,j\}$ means we have a term in our Hamiltonian acting on sites $i$ and $j$, and we sum over all such pairs such that $i < j$. In the $O(\beta)$ term, we note that the trace does not depend on $i$ or $j$ for either model in Eqs. \eqref{eq:ex_H} and \eqref{eq:SK_H}, so keeping $i$ fixed and summing over the other $N-1$ terms, we are summing over $N-1$ vectors drawn from a spherical distribution. This vanishes by the central limit theorem. This is true for each $i$, and the whole term is zero. In general, our averaging rules from the central limit theorem are:
\begin{align}\label{eq:momentum_ave_rules}
    \frac{1}{N}\sum_{i=1}^{N}v_i^{a} &= 0+O(N^{-1/2}) \,,\nonumber\\
    \frac{d}{N}\sum_{i=1}^{N}v_i^{a}v_i^{b} &= \delta^{ab}+O(N^{-1/2})\,.
\end{align}

Proceeding to order $\beta^2$, a similar argument shows that only terms with $\{i_1,j_1\}=\{i_2,j_2\}$ survive, since this ``squares'' the vectors. More generally, we have the rules:
\begin{itemize}
\item Performing the sum over sites yields a factor of $N$, which we wish to maximize.
\item Each vector associated with a site must appear at least twice, to have a non-zero average. 
\item Should a vector associated with a specific site appear more than twice, this costs a factor of $N$.
\item For a fixed set of terms in the Hamiltonian, when these operators appear in the trace at a given order of the expansion, we must sum over all possible orderings.
\end{itemize}  
These rules will enforce that the high-temperature expansion at order $L$ should be organized in terms of the conjugacy classes on the symmetric group on $L$ elements that have no fixed points, i.e., that are \emph{derangements}. This is because each vector from the couplings must appear twice, to maximize $N$ without vanishing, and via its dot product, ``points'' to the next vector, i.e., specifies a cyclic permutation. At a given order in the expansion, we then fix a set of sites. To sum over their contribution to the partition function, we must sum over all such possible disjoint products of these cyclic permutations, which is equivalent to summing over all the possible conjugacy classes and the permutations they contain. That way we cover all possibilities of how a vector might appear twice, but no more. We find:
\begin{align}\label{eq:large_N_rot_Z}
    \text{tr}[(\hat{H}^{\rm rot.})^{L}]=
      \Big(\frac{\mu}{2N}\Big)^L &\sum_{\{i_1,\cdots,i_L\}\subset\Lambda} \sum_{[\lambda_1,\cdots,\lambda_{k}]\in \derange{L}}
    \sum_{\sigma\in [\lambda_1,\cdots,\lambda_{k}]}\frac{1}{2^{k-\alpha_2}}\prod_{j=1}^{L} \vec{v}_{i_j} \cdot \vec{v}_{i_{\sigma(j)}} \tr \Big[ \prod_{j=1}^{L} \hat{O}_{i_j i_{\sigma(j)}} +\text{perms} \Big]
  \end{align}
While this formula is dense, we are simply summing over all possible subsets of size $L$ of sites of our system, and then summing over all the conjugacy classes, labeled as $[\lambda_1,\cdots,\lambda_k]$, that are in the set of the conjugacy classes $d(L)$ that are derangements, and then summing over the permutations in these classes which organize the sites that are interacting via the coupling structure. Finally, the $\alpha_2$ is just the total number of $\lambda_i$ with $\lambda_i=2$, for the given cycle structure, and we refer to \cref{sec:permutation_facts} for this alternative way of labeling conjugacy classes by the number of cycles of a given length. 
We must sum over all possible orderings of the operators in the trace, which is denoted as ``+perms'', and a more detailed description of this when $\hat{O}_{ij} = \frac{1}{2} \hat{\rho}(ij)$ can be found in \cref{sec:proving_trace_formula}. 

As an aside, all terms that survive the large-$N$ limit can be represented as a particular class of horizontal chord diagrams, as illustrated in \cref{fig:surviving_diagrams}, which we will find useful in \cref{sec:proving_trace_formula} below. In general, the bonds between horizontal lines are the two-site operators in the Hamiltonian.  When we consider the high-temperature expansion of the partition function for the Heisenberg magnet, the product of the operators from the Hamiltonian work out to be a specific permutation, and the trace evaluation of that permutation counts its cycle decomposition, see Eq. \eqref{eq:trace_rule} below taken from Ref.~\cite{Handscomb_1964}. Interestingly, the trace over the swap operator network represented in the chord diagram has an interpretation as a weight-system on braids, Refs.~\cite{1996q.alg.....7001B,Corfield:2021wwd}. In the braid interpretation of the weight system of the chord diagram, the trace counts the number of independent ropes used to form the braid (see Proposition 2.1 of Ref.~\cite{1996q.alg.....7001B}). The braid interpretation is constructed by using the swaps to define a crossing of ropes, and we further identify the beginning and end points of the horizontal lines, as if they are wrapped on a cylinder. Such an identification corresponds nicely to the trace over states: the quantum states at the end of the horizontal lines are identified and summed over. The number of independent closed loops or ropes in the braid thus formed is the number of cycles in the cycle decomposition of the permutation.

\begin{figure}
  \center\includegraphics[scale=0.6]{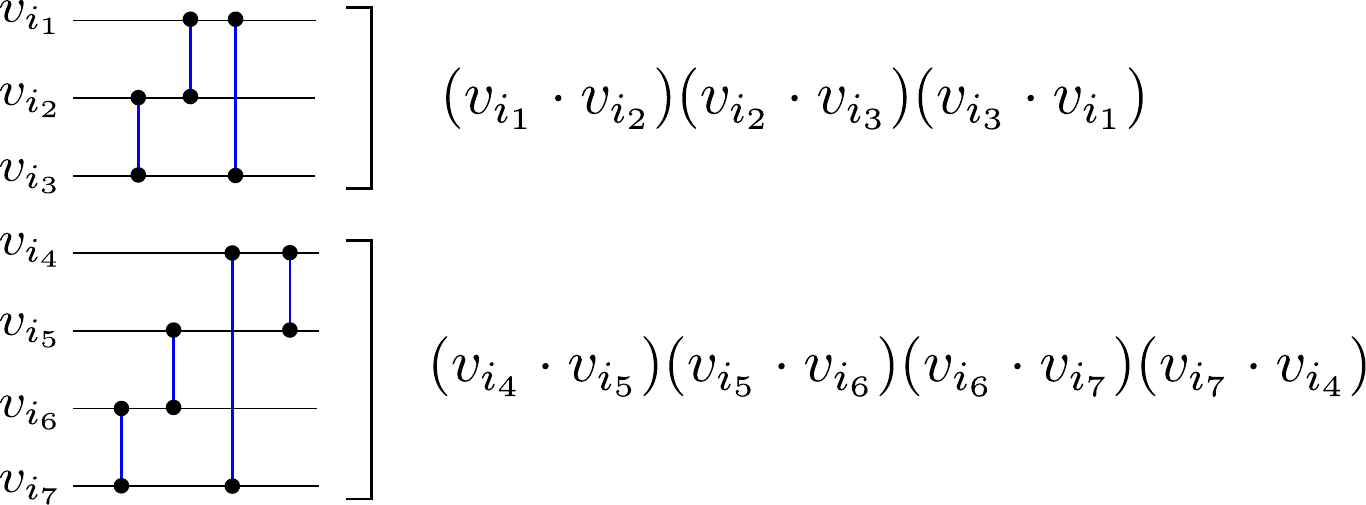}
  \caption{\label{fig:surviving_diagrams} The contribution to the partition function in the large-$N$ limit can be represented as horizontal chord diagrams. Each line in the chord diagram corresponds to the local Hilbert space of a particular spin, and the chord diagram gives the ordering of the operators we are tracing over. For the rotor-distributed couplings, only a particular class of chord diagrams can survive the large-$N$ limit. When we are considering the Heisenberg model, the operators are simply swaps, and regardless of the ordering of the swap operators in the chord diagram, the trace always evaluates to the same factor. Here we give an example contribution that would survive the large-$N$ limit after summing over groups of 7 sites, labeled by $i_1,\cdots, i_7$, which contributes at order $\beta^7$, and we indicate the coupling matrix weight that comes with the diagrams.  }
\end{figure}

\subsection{Evaluating Traces for $\hat{H}^{\rm rot.}_{\rm SK}$}
As a warm-up whose essential features carry over to the Heisenberg case, the partition function generated by Eq. \eqref{eq:SK_H} is particularly easy to compute, as all operators in the Hamiltonian commute:
\begin{align}
\hat{O}_{ij}&=\hat{\sigma}_{i}^{z}\hat{\sigma}_{j}^{z}\,,\\
2^{N-\alpha_2}\Gamma(1+L)&=\text{tr}\Big[\prod_{j=1}^{L}\hat{O}_{i_ji_{\sigma(j)}}+\text{perms}\Big]\,.\label{eq:tr_for_sk}
\end{align}  
The $2^N$ arises from the counting of states, as $\text{tr}[\mathbbm{1}]=2^N$, and we recall that each Pauli $\sigma^z_i$ is guaranteed to appear twice in the operator products, thus squaring to the identity. The factor of $2^{-\alpha_2}$ arises because of our demand that we sum over all permutations of the operators in the trace. When we have the case that the operators correspond to a two-cycle structure in the couplings, we will overcount the number of permutations, which we then correct. Counting the permutations then yields Eq. \eqref{eq:tr_for_sk}. 
After we have evaluated the trace over the operators, we can perform the sum over all the unique subsets of sites, allowing us to evaluate the sums over the vectors $\vec{v}_{i}$ in Eq. \eqref{eq:large_N_rot_Z} via the central limit theorem, \emph{before} we evaluate the sums over the permutations and the conjugacy classes. The key point is that as we have large-$N$ limit, the permutations do not depend on the specific sites being permuted, so we may evaluate the sum over sites. Each cycle in the specific permutation $\sigma$ in Eq. \eqref{eq:large_N_rot_Z} yields a factor of $d^{1-\lambda_i}$. This is because we may write:
\begin{align}
\prod_{j=1}^{L} \vec{v}_{i_j} \cdot \vec{v}_{i_{\sigma(j)}}=\vec{v}_{i_1} \cdot \vec{v}_{i_{\sigma(1)}}\times\vec{v}_{i_{\sigma(1)}} \cdot \vec{v}_{i_{\sigma^2(1)}}\times\cdots\times\vec{v}_{i_{\sigma^{\lambda_{\ell}}(1)}} \cdot \vec{v}_{i_{1}}\times\cdots
\end{align}
We have assumed the number $1$ is in the $\ell$-th cycle of $\sigma$, with length $\lambda_\ell$, and $\sigma^j$ means we have applied the permutation $j$ times. The set of vectors thus selected will not appear again in the product over the other vectors, and we may apply the central limit theorem to this set to achieve:
\begin{align}
\sum_{\{i_1,i_{\sigma(1)},\cdots,i_{\sigma^{\lambda_{\ell}}(1)}\}} \vec{v}_{i_1} \cdot \vec{v}_{i_{\sigma(1)}}\times\vec{v}_{i_{\sigma(1)}} \cdot \vec{v}_{i_{\sigma^2(1)}}\times\cdots\times\vec{v}_{i_{\sigma^{\lambda_{\ell}}(1)}} \cdot \vec{v}_{i_{1}}= \binom{N}{\lambda_{\ell}} d^{1-\lambda_\ell}\Big(1+O(N^{-1/2})\Big)\,.
\end{align}
All the other cycles in $\sigma$ can be handled similarly and independently, yielding:
\begin{align}
    \text{tr}[(\hat{H}^{\rm rot.}_{\rm SK})^{L}]=2^N\Big(\frac{\mu}{2N}\Big)^L\Gamma(1+L)\binom{N}{L}\sum_{[\lambda_1,\cdots,\lambda_{k}]\in \derange{L}} 2^{-k} \big|[\lambda_1,\cdots,\lambda_{k}] \big| \prod_{i=1}^{k}\frac{1}{d^{\lambda_i-1}}\,.
\end{align}
The binomial factor $\binom{N}{L}$ arises from counting the number of unique subsets of all sites of size $L$, and can be approximated in the large $N$ limit as
\begin{align}
  \binom{N}{L} &= \frac{N!}{(N-L)!L!} \approx \frac{N^L}{\Gamma(1+L)}
\end{align}
we have:
\begin{align}\label{eq:trace_eval_SK}
    \text{tr}[(\hat{H}^{\rm rot.}_{\rm SK})^{L}]=2^N\Big(\frac{\mu}{2d}\Big)^{L}\sum_{[\lambda_1,\cdots,\lambda_{k}]\in \derange{L}}\Big(\frac{d}{2}\Big)^{k}\Big|[\lambda_1,\cdots,\lambda_{k}]\Big|
\end{align}
We then have the partition function:
\begin{align}
Z^{\rm rot.}_{\rm SK}(\beta)&=1+\sum_{L=1}\frac{1}{\Gamma(1+L)}\Big(-\frac{\beta \mu}{2d}\Big)^{L}\sum_{[\lambda_1,\cdots,\lambda_{k}]\in \derange{L}}\Big(\frac{d}{2}\Big)^{k}\Big|[\lambda_1,\cdots,\lambda_{k}]\Big|\,.
\end{align}
We have lost the factor of $2^N$, given how we normalized the partition function. Using classic results on the combinatorics of the conjugacy classes of symmetric groups, we can recognize the generating function of the associated Stirling number of the first kind, which we reviewed in \cref{sec:gen_func_stirling}. 
Inserting $z=-\frac{\beta\mu}{2d}$ and $u=\frac{d}{2}$ into \cref{eq:gen_func_stirling}, we have:
\begin{align}
    Z^{\rm rot.}_{\rm SK}(\beta) &=\frac{\exp\big(\frac{\mu\beta}{4}\big)}{\big(1+\frac{\mu\beta}{2d}\big)^{d/2}}\Big(1+O(N^{-1/2})\Big)\,.
\end{align} 

Finally, we can perform the inverse Laplace transform to get the density of states:
\begin{align}\label{rotor_sk_density_states}
    p^{\rm rot.}_{\rm SK}(\omega) &= \frac{(2d)^{d/2}}{\Gamma(\frac{d}{2})\mu}\Big(\frac{\omega}{\mu} + \frac{1}{4}\Big)^{-1+d/2}\exp\Bigg(-2d\Big(\frac{\omega}{\mu} + \frac{1}{4} \Big)\Bigg)\Theta\Big(\frac{\omega}{\mu} + \frac{1}{4}\Big)\,.
\end{align}
\subsection{Evaluating Traces for $\hat{H}^{\rm rot.}_{q}$}
Having learned how the high-temperature expansion worked for the commuting operators, at first sight, the situation would appear hopeless for the Heisenberg Hamiltonian, since the operators in the Hamiltonian do not commute. However, we shall find a surprising result: though the operators do not commute, the large-$N$ limit selects for a particular set of operator products that evaluate to the same trace, regardless of the ordering. This will allow us to repeat the argument we gave above for the commuting operators. Focusing on the Heisenberg Hamiltonian of Eq.~\eqref{eq:ex_H}, the high-temp expansion has the form:
\begin{align}
    Z^{\rm rot.}_{q}(\beta)&=\frac{\text{tr}[\e^{-\beta\hat{H}^{\rm rot.}_{q}}]}{\text{tr}[\mathbbm{1}]}=1+\sum_{L=1}^{\infty}\frac{(-\beta)^{L}}{\Gamma(1+L)}\frac{\text{tr}[(\hat{H}^{\rm rot.}_{q})^L]}{\text{tr}[\mathbbm{1}]}.\\
   \text{tr}[(\hat{H}^{\rm rot.}_{q})^L]&= \sum_{\{i_1,j_1\}}\sum_{\{i_2,j_2\}}\cdots\sum_{\{i_L,j_L\}}\Big(\prod_{k=1}^{L}(\frac{\mu}{2N}\vec{v}_{i_{k}}\cdot\vec{v}_{j_{k}})\Big)\times\text{tr}[\exch{i_1 j_1}\cdot \exch{i_2 j_2}\cdot\,...\,\cdot \exch{i_L j_L}]
\end{align}

Note $\exch{i_1 j_1}\cdot \exch{i_2 j_2}\cdot\,...\,\cdot \exch{i_L j_L}\in\Sym{N:2L}$, with $\Sym{N:2L}$ the set of permutations on $N$ objects, with at most $2L$ objects being non-trivially permuted. $\Sym{N}=\Sym{N:N}$ is the symmetric group. The operator product can be given a disjoint cyclic decomposition, standard for permutations:
\begin{align}\label{eq:product_ex.}
\exch{i_1 j_1}\cdot \exch{i_2 j_2}\cdot\,...\,\cdot \exch{i_L j_L}=(a_{11}a_{12}\cdots a_{1 \ell_{1}})(a_{21}a_{22}\cdots a_{2 \ell_{2}})\cdots(a_{k1}a_{k2}\cdots a_{k \ell_{k}})
\end{align}  
The length of the $i$-th cycle permutation is $\ell_i$, and for each $i$, $1\leq\ell_{i}\leq \ell_{i+1}\leq N$ (the cycle decomposition of the product of swap operators is in principle distinct from the decomposition we found for the couplings, thus our use of $\ell$, rather than $\lambda$). We do not suppress the trivial one-cycles. Each number $a_{mn}\in \{1,\cdots,N\}$ labels a unique site in the system. The cycle operator $(a_{i1}\cdots a_{i\ell_{i}})$ has the following action on the product state over the sites at $a_{i1},\cdots, a_{i\ell_{i}}$:
\begin{align}
  (a_{i1}\cdots a_{i\ell_{i}})&|z_{1}\rangle_{a_{i1}}\otimes\cdots\otimes |z_{\ell_i}\rangle_{a_{i\ell_i}}\otimes |\psi\rangle\nonumber\\
  &=|z_{\ell_i}\rangle_{a_{i1}}\otimes |z_{1}\rangle_{a_{i2}}\otimes \cdots\otimes |z_{\ell_{i-1}}\rangle_{a_{i\ell_i}}\otimes |\psi\rangle\,.
\end{align}  
$|\psi\rangle$ is an arbitrary state of all the other sites. 

In the following, in order to keep track of the contribution from the number of flavors, we keep its value $n_f$ general, and $n_f=2$ corresponds to the original Heisenberg interaction. The conjugacy class of the product in Eq.~\eqref{eq:product_ex.} is $[\ell_1,\cdots,\ell_k]$, and we have for the trace evaluation (Ref.~\cite{Handscomb_1964}):
\begin{align}\label{eq:trace_rule}
\ln\Big(\text{tr}[\exch{i_1 j_1}\cdot \exch{i_2 j_2}\cdot\,...\,\cdot \exch{i_L j_L}]\Big)&=\Big(N+k-\sum_{i=1}^{k}\ell_i\Big)\ln n_{f}=k \ln n_{f}
\end{align}
That is, the trace only depends on the number of disjoint cycles in the product of the exchanges. This trace rule is easily derived from the fact that all sites being acted on in the same cycle must be in the same state to survive the trace. 

To finish calculating the partition function, we find the surprising fact that the exchange operators effectively commute under trace to the leading order in $N$. That is, the fundamental formula we need is:
\begin{align}\label{eq:trace_rule_for_cycles}
    \tr \big[\exch{i_1 i_2}\cdot \exch{i_2 i_3}\cdot\,...\cdot \exch{i_\lambda i_{1}}+\text{perms}\big]&=\Gamma(1+\lambda)n_{f}^{N+2-\lambda},\text{ for } \lambda>2\, ,
\end{align}  
where ``+perms'' means summing over all orderings of the exchange operators. We give a proof of this fact in \cref{sec:proving_trace_formula}, but here we satisfy ourselves with an illustrative example:
\begin{align}
\text{tr}\big[(12)(23)(34)(45)(15)\big]&=\text{tr}\big[(2345)(1)\big]=n_{f}^{N-5+1+1}\\
\text{tr}\big[(12)(23)(34)(15)(45)\big]&=\text{tr}\big[(15)(234)\big]=n_{f}^{N-5+1+1}
\end{align}  
That is, the first permutation has $N-5$ one-cycles, from the $N-5$ elements not involved in the permutation. The product evaluates to 1 4-cycle and 1 1-cycle. In the second case, we again have $N-5$ 1-cycles, but now the product evaluates to 1 2-cycle and 1 3-cycle. Importantly, the number of total cycles remained the same under commutation of the exchange operators. In general, when we consider the substitution of the Heisenberg (swap) interaction operators into Eq. \eqref{eq:large_N_rot_Z}, we wish to evaluate traces of the form:
{\small\begin{align}\label{traces_in_heisen_model}
    \tr \Big[\big(\exch{i_1 i_2}\cdot \exch{i_2 i_3}\cdot\,...\cdot \exch{i_{\lambda_1} i_{1}}\big)\big(\exch{i_{\lambda_{1}+1} i_{\lambda_{1}+2}}\cdot \exch{i_{\lambda_{1}+2} i_{\lambda_{1}+3}}\cdot\,...\cdot \exch{i_{\lambda_1+\lambda_2} i_{\lambda_1+1}}\big) \cdots +\text{perms} \Big]&=2^{-\alpha_2}\Gamma(1+L)n_{f}^{N}\prod_{j=1}^{k}n_{f}^{2-\lambda_j}
  \end{align}  }
Critically, because of the effective commutation of the swaps in the trace, the trace evaluates to a function that only depends on the conjugacy class enforced by the couplings. Using this formula in Eq.~\eqref{eq:large_N_rot_Z}, performing the sums over the momenta via the central limit theorem, and inserting $z=-\frac{\beta\mu}{2n_f d}$ and $u=\frac{n_f^2d}{2}$ into \cref{eq:gen_func_stirling} we can immediately deduce the partition function:
\begin{align}\label{eq:ex_part}
    Z^{\rm rot.}_{q}(\beta)\Big|_{\rm ann.}&= \frac{\exp\big(\frac{\mu n_{f}\beta}{4}\big)}{\big(1+\frac{\mu\beta}{2n_{f}d}\big)^{dn_{f}^2/2}}\Big(1+O(N^{-1/2})\Big)\,.
\end{align}  
This predicts a related density of states to the commuting Hamiltonian in Eq. \eqref{rotor_sk_density_states}:
\begin{align}\label{rotor_ex_density_states}
    p^{\rm rot.}_{\rm q}(\omega) &= \frac{(2n_fd)^{n_f^2d/2}}{\Gamma\big(\frac{ n_f^2d}{2}\big)\mu}\Big(\frac{\omega}{\mu} + \frac{n_f}{4}\Big)^{-1+n_f^2d/2}\exp\Bigg(-2dn_f\Big(\frac{\omega}{\mu} + \frac{n_f}{4} \Big)\Bigg)\Theta\Big(\frac{\omega}{\mu} + \frac{n_f}{4}\Big)\,.
\end{align}
We have appropriately normalized the density. We may check Eq.~\eqref{rotor_ex_density_states} by diagonalizing the Heisenberg Hamiltonian of Eq.~\eqref{eq:ex_H} numerically, and calculate its density of states, then averaging these densities over many realizations of the coupling. Fig.~\ref{fig:rotor_q} gives the result for $N=12$ sites, averaged over 2400 realizations, with momenta in $d=3,5,10$ dimensions, $n_f=2$, and $\mu=2$.

\begin{figure}
  \center\includegraphics[scale=0.5]{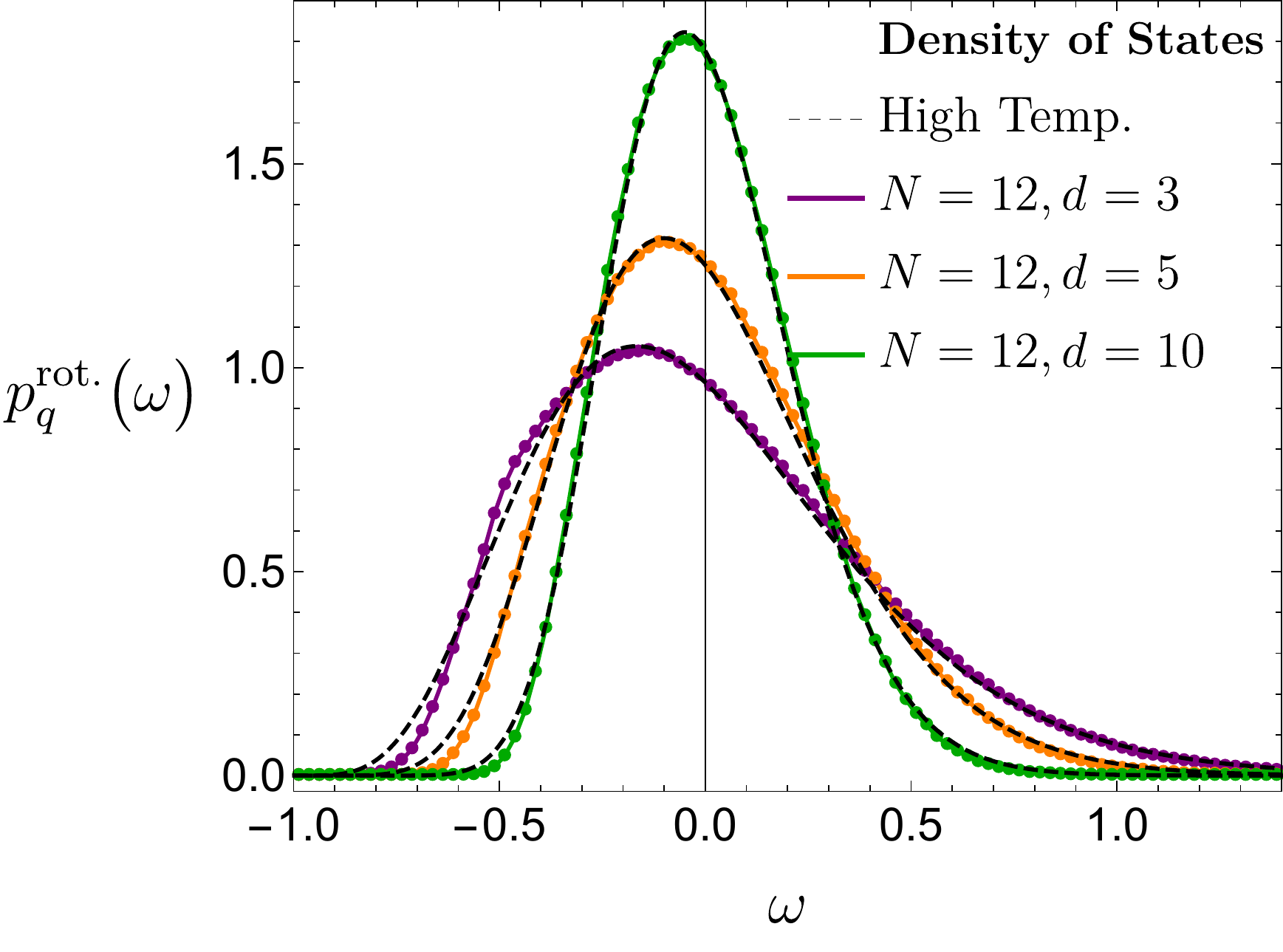}
  \caption{\label{fig:rotor_q} Numerically determined average density of states for $\hat{H}^{\rm rot.}_{\rm q}$ of Eq.~\eqref{eq:ex_H}, $\mu=2$, for $N=12$ sites, averaged over 2400 realizations, with momenta in $d=3,5,10$ dimensions. The black dashed line is Eq.~\eqref{rotor_ex_density_states}. }
\end{figure}

To contrast the derived density of states in Eqs.~\eqref{rotor_sk_density_states} and \eqref{rotor_ex_density_states}, it is interesting to compare to the results of Refs.~\cite{MON197590,Baldwin:2019dki}, and more rigorously, \cite{2014MPAG...17..441E}. There, they examine generic all-to-all spin models or $k$-body fermionic interactions with Gaussian distributed couplings. When interactions are few-body, the resulting density of states was proved to be Gaussian distributed itself. The models studied in those papers have no assumptions about $\SU(n_f)$ symmetry, but this does not appear to be the necessary ingredient to yield a non-Gaussian density of states, as can be checked from the high-temperature calculation of the free energy of spin-glass models found in Ref.~\cite{A_J_Bray_1980}. Indeed, it appears that we need the combination of the $\SU(n_f)$ symmetry and the underlying simplicity of the coupling matrix to get significant departures from Gaussian. Looking ahead to the results of Eq.~\eqref{eq:partition_classical_rotor}, we can conjecture the needed simplicity in the coupling matrix, is to require only a few non-vanishing eigenvalues. We can find some confirmation of this by examining the $d\rightarrow\infty$ limit of Eq.~\eqref{rotor_sk_density_states}. In this limit, we are increasing the number of non-trivial eigenvalues, and the maximum in the density of states begins to take a Gaussian shape, similar to the emergence of an approximate Gaussian using Laplace's method, see the $d=10$ plot in Fig. \ref{fig:rotor_q}.

\section{High-Temperature Expansion For Dynamical Rotors}\label{sec:classical}
We have seen how when we let the couplings in the quantum models be determined by unit vectors uniform on a sphere, we could calculate the partition functions in the large-$N$ limit. The self-averaging of the high-temperature expansion meant that the couplings themselves acted as if they were dynamical variables, selecting a set of terms whose structure was simply organized by the structure of conjugacy classes of the symmetric group. This observation begs the question, to what extent will our methods work for classical rotor models, where we replace our quantum operators with classical spins, and generalize our coupling matrix to be arbitrary? We are considering the case when we have the classical Hamiltonian of Eq.~\eqref{eq:_dyn_cl_rotor_H},
\begin{align}
    H_{\rm cl.}(J) = \frac{1}{2}\sum_{i,j=1}^N J_{ij}\vec{S}_{i}\cdot\vec{S}_j = \sum_{1\leq i < j \leq N}J_{ij}\vec{S}_{i}\cdot\vec{S}_j + \frac{1}{2} \tr[J]\,,
\end{align}  
and we are calculating the partition function as an average over the $\vec{S}_{i}$ on the unit sphere embedded in $\rotdim$-dimensions. The assumption on the coupling matrices that we will need for our high-temperature expansion is the following. We are given a list of real numbers $e_{\ell}, \ell=1,\cdots,M$, and we are given a generic orthogonal $N\times N$ matrix $O_{ij}$. With these, we calculate the coupling matrix of our Hamiltonian as:
\begin{align}
J_{ij}=&\sum_{k=1}^{M}O_{ik}O_{jk}e_{k},\text{ and } \big|O_{ij}\big| \leq O_{\rm max}\sqrt{\frac{\log(N)}{N}},\,\,\forall i,j \,,\ \nonumber\\
&\text{ with }M\ll N,\text{ and }\sum_{k=1}^{N}O_{ik}O_{jk}=\delta_{ij}\,.
    \label{eq:shift_rule_for_matrix}
\end{align}
For instance, random orthogonal matrices $O_{ij}$ sampled from the Haar distribution on $O(N)$ will satisfy this condition for $O_{max}=\sqrt{6}$ almost surely for large $N$~\cite{Jiang2005}. Note that the diagonals of the coupling matrix do not contribute to the physics, since they correspond to a constant shift of the Hamiltonian. Hence the physics of the partition function can only depend on the off-diagonal entries of the coupling matrix. However, as we shall see, the \emph{full} matrix $J_{ij}$ which has the eigenvalues $e_{\ell}$ will determine the structure of the partition function in the high-temperature regime. To perform the high-temperature expansion, we write:
\begin{align}
    \exp\big(-\beta H_{\rm cl.}(J)\big) = \exp\Big(-\frac{1}{2}\sum_{i=1}^{N}\beta J_{ii}\Big)\times \exp\Big(-\beta\sum_{\{i,j\}}J_{ij}\vec{S}_{i}\cdot\vec{S}_j\Big)\,,
\end{align}
and we need only expand the non-trivial term $\exp\big(-\beta\sum_{\{i,j\}}J_{ij}\vec{S}_{i}\cdot\vec{S}_j\big)$. Alternatively, we multiply the partition function by the factor $\e^{\frac{1}{2}\beta\text{tr}[J]}$.

The calculation of the partition function repeats the arguments from \cref{sec:annealed}, but the couplings playing the role of the quantum operators, and the $\vec{v}_i$ becoming $\vec{S}_i$, now with an explicit integral over the rotors, rather than invoking the central limit theorem at large $N$. We achieve:
\begin{align}\label{eq:large_N_rot_Z_cl}
    \frac{(H_{\rm cl.})^{L}}{\Gamma(1+L)}&= \sum_{\{i_1,\cdots,i_L\}\subset\Lambda} \sum_{[\lambda_1,\cdots,\lambda_{k}]\in \derange{L}}\sum_{\sigma\in [\lambda_1,\cdots,\lambda_{k}]}\frac{1}{2^k} \prod_{j=1}^{L} \Big( \vec{S}_{i_j} \cdot \vec{S}_{i_{\sigma(j)}} J_{i_ji_{\sigma(j)}} \Big)\,.
\end{align}

Integration over the $\vec{S}_i$ will again project onto derangements of the indices. So performing the integration, we get: 
\begin{align}\label{eq:large_N_rot_Z_cl_II}
 \e^{\frac{1}{2}\beta\text{tr}[J]}Z_{\rm cl.}(\beta) = 1 + \Big(-\frac{\beta}{\rotdim}\Big)^L \sum_{\{i_1,\cdots,i_L\}\subset\Lambda} \sum_{[\lambda_1,\cdots,\lambda_{k}]\in \derange{L}}\sum_{\sigma\in [\lambda_1,\cdots,\lambda_{k}]} \Big(\frac{\rotdim}{2}\Big)^k \prod_{j=1}^{L} J_{i_ji_{\sigma(j)}}
\end{align}

We note that the sums over the subsets of sites can be moved past the permutations, as the permutation only acts on the sub-index, and leaves the subset of sites we are considering invariant. We then perform the sum over the sites:
\begin{align}\label{eq:trace_rule_start}
\sum_{\{i_{1},\cdots,i_{L}\}\subset\Lambda}\prod_{j=1}^{L} J_{i_ji_{\sigma(j)}}=f_{[\lambda_1,\cdots,\lambda_{k}]}(J)\,.
\end{align}
That is, the sum over all unique subsets of the sites of size $L$ will no longer depend on the specific derangement, but just its conjugacy class. We would like to determine the function $f$, at least in the large-$N$ limit. First, we note that if we summed over all the site indices $1,\cdots, N$, without the uniqueness restriction (and we explicitly include the diagonal terms in $J$ if non-zero), we would have:
\begin{align}\label{eq:trace_rule}
    f_{[\lambda_1,\cdots,\lambda_{k}]}(J)&\approx\prod_{j=1}^{k}\tr[J^{\lambda_j}]\,.
\end{align}
Further, if $L=\lambda_1+\lambda_2+\cdots +\lambda_k\ll N$, and one attempted to use a Monte Carlo method to evaluate the trace in Eq. \eqref{eq:trace_rule}, then on average we would have no repeated sites in each sample we take (assuming a flat probability distribution over sites), which would lead to the sum in Eq. \eqref{eq:trace_rule_start} on average. Ultimately, we are summing over unique sites, this can only be approximately true. The estimated error on this approximation is:
\begin{align}
\label{eq:trace_rule_with_error}
 f_{[\lambda_1,\cdots,\lambda_{k}]}(J)&=\prod_{j=1}^{k}\tr[J^{\lambda_j}]\Big(1+O\big(L\frac{M\log(N)}{N}\big)\Big)\,.
\end{align}
We refer the reader to \cref{sec:inequalities} for a more detailed analysis of this claim, where we show this condition follows from our assumption of Eq. \eqref{eq:shift_rule_for_matrix}. With this approximation in place, we can perform an evaluation of the partition function:
\begin{align}
  \e^{\frac{1}{2}\beta\tr[J]}Z_{\rm cl.}(\beta)&=1+\sum_{L=2}\frac{1}{\Gamma(1+L)}\Big(-\frac{\beta}{\rotdim}\Big)^{L}\sum_{[\lambda_1,\cdots,\lambda_{k}]\in \derange{L}}\Big(\frac{\rotdim}{2}\Big)^{k}\Big|[\lambda_1,\cdots,\lambda_{k}]\Big|\prod_{i=1}^{k}\text{tr}[J^{\lambda_{i}}], \nonumber\\
  &= \exp\Big( \tr\big[\frac{\beta}{2} J -\frac{\rotdim}{2}\ln(1+\frac{\beta}{\rotdim}J)\big] \Big) = \prod_{\ell=1}^M \big(1 +\frac{\beta}{\rotdim} e_\ell \big)^{-\frac{\rotdim}{2}} \e^{\frac{\beta}{2} e_\ell} , \label{eq:partition_classical_rotor}
\end{align}
that is
\begin{align}
  Z_{\rm cl.}(\beta) &= \prod_{\ell=1}^M \big(1 +\frac{\beta}{\rotdim} e_\ell \big)^{-\frac{\rotdim}{2}} , \label{eq:partition_classical_rotor}
\end{align}
where the summation was accomplished with the help of \cref{eq:matrix-gen-fun} in \cref{sec:permutation_review}. The radius of convergence then is set by the largest eigenvalue of the coupling matrix, i.e., $|\beta| < d_r (\max_{\ell}{e_\ell})^{-1}$. Moreover, we see the role of the trace of the coupling matrix: it represents the generator of one-cycle contributions to the counting of permutations, and cancels the exponential factor generated by the other contributions. This representation of the partition function follows from substituting into the moment generating function for the coupling matrix $J$, and the expression of the traces in terms of the eigenvalues of $J$. 

\subsection{Lower Bound for Ground State Energy of Classical Rotors with a Positive Pole Structure}
If the locations of all the poles of \cref{eq:partition_classical_rotor} are positive, then we can find a lower bound to the ground state energy. While it may seem that we can change the values of the eigenvalues of the coupling matrix $J$ by redefining the diagonal of $J$, if $J$ satisfies \cref{eq:shift_rule_for_matrix}, the locations of the \emph{poles} remain unchanged. First, we can compute the density of states as: 
\begin{align}
   \rho(\omega) = \frac{1}{2\pi i} \int\displaylimits_{c-i\infty}^{c+i\infty}\d\beta e^{\omega\beta} Z_{\rm cl.}(\beta) &= \frac{1}{2\pi i} \int\displaylimits_{c-i\infty}^{c+i\infty}\d\beta\exp\Bigg(\omega\beta - \sum_{\ell=1}^{M} \frac{\rotdim}{2}\ln\Big(1+\frac{\beta  e_\ell}{\rotdim}\Big) \Bigg)\,.
\end{align}
We want to evaluate the density of states as we move the real part of $\beta$ to infinity. This selects the low-temperature regime, which should give the density of states near the ground state. This is equivalent to moving the contour of integration to the right in the complex plane, so sending $c\rightarrow\infty$ in the limits of integration. In particular, we evaluate the density of states at the energy:
\begin{align}\label{eq:groundstate}
    \omega_0& = 0.
\end{align}
This yields the integral:
\begin{align}
   \rho(\omega_0) &= \frac{1}{2\pi i} \int\displaylimits_{c-i\infty}^{c+i\infty}\d\beta\exp\Bigg(-\frac{\rotdim}{2}\sum_{\ell=1}^{M}\ln\Big(1+\frac{\beta e_\ell}{\rotdim}\Big)\Bigg)\,,\\
   &=\frac{1}{2\pi } \int\displaylimits_{-\infty}^{\infty}\d z\exp\Bigg(-\frac{\rotdim}{2}\sum_{\ell=1}^{M}\ln\Big(1+\frac{e_\ell}{\rotdim}(c+iz)\Big)\Bigg)\,.\label{eq:find_ground_state}
\end{align}
Now along the integration contour, the real part of the log is:
\begin{align}
\text{Re }\Big\{\ln\Big(1+\frac{e_{\ell}}{\rotdim}(c+i z)\Big)\Big\}&=\frac{1}{2}\ln\Big(1+2\frac{c e_{\ell}}{\rotdim}+\frac{c^2 e^2_{\ell}}{\rotdim^2}+\frac{e^2_{\ell}z^2}{\rotdim^2}\Big)
\end{align}
This means that the real part of the argument of the exponential in Eq.~\eqref{eq:find_ground_state} is negative on the whole contour of integration, and can be arbitrarily large by shifting $c\rightarrow\infty$. We conclude:
\begin{align}
\rho(\omega_0)&=0\,.
\end{align}
We can of course extend this argument to any energy less than $\omega_0$, concluding:
\begin{align}
\rho(\omega)&=0\,\text{ if } \omega\leq\omega_0.
\end{align}
Thus $\omega_0=0$ is a lower bound on the ground state energy. 
This lower bound agrees with using a convex relaxation argument when $J$ is positive semi-definite, as will be detailed in \cref{sec:classical-analysis}.

\subsection{Numerical Check}
We give a numerical check for our formula \eqref{eq:partition_classical_rotor} by calculating the energy of the system using the standard heat bath algorithm for when $\rotdim = 3$. We consider the classical neutrino problem\footnote{a more complicated problem illustrating the necessity of the assumption that the eigenvalues do not scale extensively with system size can be found in \cref{sec:breakdown}}, and take $\mu_1= -\mu_2 = 4$ so that the coupling matrix is $J_{ij}= \frac{1}{N}(1  - \vec{v}_i\cdot\vec{v}_j)$ with $\vec{v}_i$ uniformly distributed on the 2-sphere,\footnote{Though essentially any distribution of momenta can be chosen.} and the neutrino flavors taken to be $2$, so that the classical dynamical rotor variables are embedded in $\rotdim=3$ dimensions. The eigenvalues of $J_{ij}$ in the $N \rightarrow \infty$ limit are:
\begin{align}
\{e_{\ell}\}_{\ell=1}^{N}=\Big\{1, -\frac{1}{3}, -\frac{1}{3}, -\frac{1}{3}, 0, 0, \cdots\Big\}\,.
\end{align}
This yields from Eq.~\eqref{eq:partition_classical_rotor} the partition function:
\begin{align}\label{eq:partition_classical_neutrino}
Z_{\rm cl.}\big(\beta;1,1,\{\vec{v}_i\}\big)=\Big(1-\frac{\beta}{9}\Big)^{-9/2}\Big(1+\frac{\beta}{3}\Big)^{-3/2}\,.
\end{align}
Calculation of the energy is then straightforward. Further, we can solve for the energy using the heat bath algorithm, and compare it to the analytic expression, yielding Fig.~\ref{fig:free_e_cl_neutrino}. For a given $N$, the energy is determined from a single realization of the couplings. We see from the partition function that we can expect a phase transition at $\beta=9$, which formally can happen only at infinite $N$. Thus we calculate from the heat bath algorithm for multiple $N$, and use a simple form of the $N$ dependence at each temperature to extrapolate to infinite $N$:
\begin{align}
-\frac{\partial}{\partial\beta}\ln Z(N)=a(\beta)+\frac{b(\beta)}{N}+\frac{c(\beta)}{N^2}\,.
\end{align}
Then taking $N\rightarrow\infty$ limit yields the extrapolated curve.

\begin{figure}
  \center\includegraphics[scale=0.5]{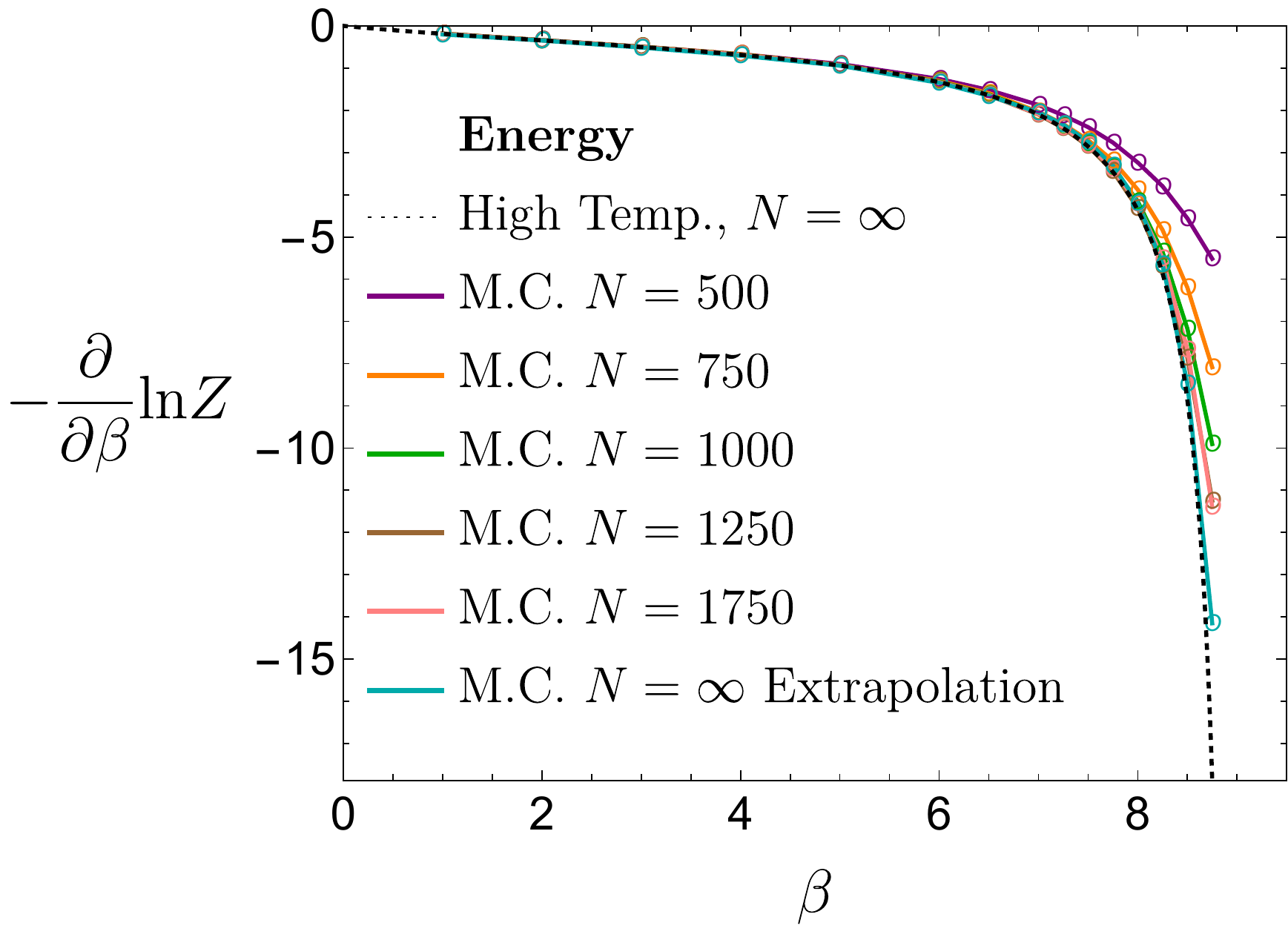}
  \caption{\label{fig:free_e_cl_neutrino} The numerically determined energy from the heat bath algorithm. For a given $N$, the free energy is determined from a single realization of the couplings. The circles are the determination of the free energy at a specific temperature, lines are linear interpolations. The statistical uncertainty is smaller than can be shown. The black dashed line is the determination of the free energy from the high-temperature expansion from Eq.~\eqref{eq:partition_classical_neutrino}.}
\end{figure}

\section{Phase Transitions and Flavor-Momentum Locked State in the Neutrino Hamiltonian}\label{sec:FML}

Our calculated partition function in the large-$N$ limit for the classical neutrino problem, \cref{eq:partition_classical_neutrino}, predicts a pole at the temperature $\beta=9$, signaling a phase-transition\footnote{It is easy to convince oneself that the other pole at $\beta=-3$ would correspond to a ferromagnetic state where all spins align.}. From the Leib inequalities, Eq.~\eqref{eq:Leib_ineq}, this implies that the quantum partition function must also diverge and undergo a phase transition. Here we elucidate the nature of this new phase, which we term the flavor-momentum locked (FML) state. 

We can also directly verify that the quantum system undergoes a phase transition with Eq.~\eqref{eq:ex_part}. In the coupling matrix for the generalized neutrino Hamiltonian of Eq.~\eqref{eq:neutrino_H}, we will see that it is inconsequential to the FML state whether $\mu_1=0$ or not. Thus the $\mu_1=0$ Hamiltonian whose partition function is calculated in the high-temperature regime in Eq.~\eqref{eq:ex_part} undergoes a phase-transition to the same FML state as the standard quantum neutrino Hamiltonian.

In what follows, we solve the full spectrum of the classical neutrino Hamiltonian, which allows us to determine the full phase diagram. This is then followed by an analysis of the classical equations of motion, to establish that the coherent-state path-integral of the quantum problem will have the same saddle-point. Moreover, we argue that as $N\rightarrow \infty$, the classical solution is in fact stable and bounded when we are sufficiently close to the exact FML state, which we check numerically.

\subsection{Classical analysis}\label{sec:classical-analysis}
The classical neutrino partition function exhibits a phase transition at $\beta=9$, and we would like to understand the phase of matter on the other side of the transition, i.e., the low-temperature regime. Recall that the classical neutrino Hamiltonian is:
\begin{align}
    H^{\nu\nu}_{\rm cl.}=\sum_{\{i,j\}}\frac{1}{2N}\big(\mu_1 + \mu_2\vec{v}_i\cdot\vec{v}_j\big)\vec{S}_i\cdot\vec{S}_j\,. \tag{re \ref{eq:Hnu-cl}}
\end{align}
In what follows, we will assume that the classical rotors and the dimension of the momenta in the coupling matrix are the same:
\begin{align}
d=\rotdim=3\,.
\end{align}

When $\beta \rightarrow \infty$, the partition function is dominated by the ground state. Since the rotors satisfy the unit norm constraint $\vec{S}_i\cdot\vec{S}_i = 1$ for each $i$, finding the ground state can be viewed as an optimization problem with constraints. Introducing the Lagrange multipliers, we have
\begin{align}
    H_\lambda := \sum_{\{i,j\}} J_{ij} \vec{S}_i\cdot\vec{S}_j - \frac{1}{2} \sum_i \lambda_i \big( \vec{S}_i\cdot\vec{S}_i - 1\big),
\end{align}
where we write the coupling matrices back as $J_{ij}$ to make the calculation look more general and also to simplify the notation. Taking derivatives with respect to $\vec{S}_i$, we have
\begin{align}\label{eq:lag-der}
    \sum_{j} J_{ij} \vec{S}_j - \lambda_i \vec{S}_i = 0.
\end{align}

In general, this equation is hard to solve. Here we make an assumption that all $\lambda_i$ are equal, i.e., $\lambda_i = \lambda$. This assumption is equivalent to a convex relaxation of the constraints: the $N$ independent constraints $\vec{S}_i\cdot\vec{S}_i = 1$ are relaxed to a single convex constraint $\sum_i\vec{S}_i\cdot\vec{S}_i = N$. (See \cite{Aizenman:1987gp} for a similar treatment of the SK model.) Solutions under this convex constraint give a lower bound of the ground state energy. However, as we will see soon, when $J_{ij} = \frac{1}{2N} (\mu_1 + \mu_2\vec{v}_i\cdot\vec{v}_j)$, there always exist solutions of this convex constraint that also satisfy the original constraint, and these solutions are the ground states.

Under the assumption $\lambda_i = \lambda$, \cref{eq:lag-der} becomes an eigenvalue equation 
\begin{align}
    \sum_{j} J_{ij} \vec{S}_j - \lambda \vec{S}_i = 0,
\end{align}
where each of the three $S_i^a: a=1,2,3$ is an eigenvector of the coupling matrix $J_{ij}$, which gives the saddle points. The minimum is reached when $S_i^a: a=1,2,3$ are eigenvector(s) associated with the smallest eigenvalue of $J_{ij}$. The energy associated with the eigenvalue $\lambda$ is given by
\begin{align}
    E_\lambda = \frac{1}{2} \sum_{i,j} J_{ij} \vec{S}_i\cdot\vec{S}_j \Big|_\lambda = \frac{1}{2} \lambda \sum_{i} \vec{S}_i\cdot\vec{S}_i = \frac{1}{2} N \lambda.
\end{align}

Now let's focus on the case $J_{ij} = \frac{1}{2N} (\mu_1 + \mu_2\vec{v}_i\cdot\vec{v}_j)$. Recall that in the large $N$ limit, we have 
\begin{align}
    \frac{1}{N}\sum_{i=1}^{N}v_i^{a} &= 0+O(N^{-1/2})\,,\nonumber\\
    \frac{d}{N}\sum_{i=1}^{N}v_i^{a}v_i^{b} &= \delta^{ab}+O(N^{-1/2}) \,. \tag{re \ref{eq:momentum_ave_rules}}
\end{align}
This implies that $\frac{1}{N} \mathbbm{1} +\frac{3}{N} \sum_{a} v_i^a v_j^a$ is a projector to the 4D subspace $\mathrm{span}\{\vec{1}, v_i^a : a = 1, 2, 3\}$. Therefore, $\frac{1}{2N} (\mu_1  + \mu_2\vec{v}_i\cdot\vec{v}_j)$ has 3 eigenvectors living in $\mathrm{span}\{v_i^a : a = 1, 2, 3\}$ associated with the eigenvalue $\mu_2/6$, 1 constant eigenvector $\vec{1}$ associated with the eigenvalue $\mu_1/2$, and $N-4$ eigenvectors that are orthogonal to $\mathrm{span}\{\vec{1}, v_i^a : a = 1, 2, 3\}$, and are associated with eigenvalue $0$.

Now we come back to check whether these eigenvectors can satisfy the original constraint $\vec{S}_i\cdot\vec{S}_i = 1$ for each $i$. For the eigenvectors living in $\mathrm{span}\{v_i^a : a = 1, 2, 3\}$, we can choose $\vec{S}_i = \vec{v}_i$, then these constraints are clearly satisfied. Moreover, if we perform a global $\O(3)$ rotation, these vectors remain in $\mathrm{span}\{v_i^a : a = 1, 2, 3\}$ and are still saddle points. Since in these solutions $\vec{S}_i$ are always equal to $\vec{v}_i$ up to a global rotation, this is why we refer to these states as \emph{flavor-momentum locked} (FML) states. They are ordered states, and the order parameter is defined as
\begin{align}
    M = \text{det}(M^{ab}), \text{ where } M^{ab}=\frac{3}{N}\sum_{i=1}^{N}v_{i}^{a}S_{i}^{b}\,. \label{eq:classical_order_param}
\end{align}

For the eigenvectors living in $\mathrm{span}\{\vec{1}\}$, we can make $\vec{S}_1$ satisfy the constraint, and all other constraints are automatically satisfied. These states are just the usual ordered states. For the eigenvectors living in the remaining $N-4$ dimensional eigen-subspace, we can always choose 3 vectors to satisfy the constraint. In fact, if we choose $\vec{S}_i$ randomly on $\mathbb{S}^2$,
it will be an eigenvector of the coupling matrix with eigenvalue $0$ in the large $N$ limit.

At different values of $\mu_1$ and $\mu_2$, different eigenvectors become the ground state. When both $\mu_1 \geq 0$ and $\mu_2 \geq 0$, the $0$-eigen subspace dominates. Therefore the ground states are completely disordered, and the system is in the paramagnetic phase. When $\mu_1 < 0$ and $3\mu_1 < \mu_2$, the $\mu_1/2$-eigen subspace dominates, and the system is ordered and in the ferromagnetic phase. When $\mu_2 < 0$ and $3\mu_1 > \mu_2$, the $\mu_2/6$-eigen subspace dominates, and the system is in the FML phase. The zero-temperature phase transitions are expected to be first order, because the ground states in different phases are orthogonal to each other. We summarize these results in the phase diagram \cref{fig:phase-diagram}. 

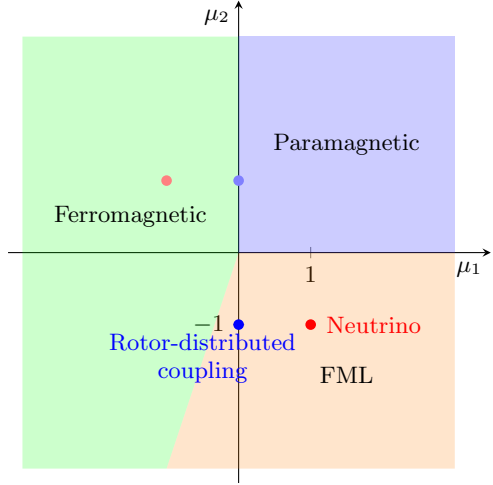
\begin{figure}
    \centering
    \begin{tikzpicture}
    \begin{axis}[
        axis lines=middle,
        xlabel={$\mu_1$},
        ylabel={$\mu_2$},
        xlabel style={at={(ticklabel* cs:1)}, anchor=north east},
        ylabel style={at={(ticklabel* cs:1)}, anchor=north east},
        xmin=-3.2, xmax=3.5,
        ymin=-3.2, ymax=3.5,
        xtick={1},
        ytick={-1},
        width=8cm, height=8cm
    ]
    
    \fill[fill=blue, opacity=0.2] (axis cs:0,0) rectangle (axis cs:3,3);
    \node at (axis cs:1.5,1.5) {Paramagnetic};

    \fill[fill=green, opacity=0.2] (axis cs:0,0) -- (axis cs: -1,-3) -- (axis cs: -3,-3) -- (axis cs: -3,3) -- (axis cs: 0,3) -- cycle;
    \node at (axis cs:-1.5,0.5) {Ferromagnetic};
    
    \fill[fill=orange, opacity=0.2] (axis cs:0,0) -- (axis cs: 3,0) -- (axis cs: 3,-3) -- (axis cs: -1,-3) -- cycle; 
    \node at (axis cs:1.5,-1.7) {FML};
    
    \fill[red] (axis cs:1,-1) circle (2pt);
    \fill[red!50] (axis cs:-1,1) circle (2pt);
    \node[anchor=west, red] at (axis cs:1.1,-1) {Neutrino};
    
    \fill[blue] (axis cs:0,-1) circle (2pt);
    \fill[blue!50] (axis cs:0,1) circle (2pt);
    \node[anchor=north, blue, align=center] at (axis cs:-0.5,-1) {Rotor-distributed\\coupling};
    \end{axis}
\end{tikzpicture}
    \caption{Zero-temperature phase diagram of the classical rotor Hamiltonian \cref{eq:Hnu-cl} in the $\mu_1$-$\mu_2$ plane. When $\mu_1 \geq 0$ and $\mu_2 \geq 0$, the ground states are completely disordered, and the system is in the paramagnetic phase. When $\mu_1 < 0$ and $3\mu_1 < \mu_2$, the ground state is ordered and the system is in the ferromagnetic phase. When $\mu_2 < 0$ and $3\mu_1 > \mu_2$, the system is in the FML phase. The dark blue and red dots denote the locations of the rotor-distributed coupling and neutrino coupling respectively. The lighter blue and red dots denote the location of corresponding couplings with signs flipped, which is equivalent to the same model but with the sign of $\beta$ flipped.}
    \label{fig:phase-diagram}
\end{figure}

It is interesting to see how this phase diagram is consistent with our high-temperature analysis. When $\mu_1=0$ and $\mu_2=-1$ (dark blue dot in \cref{fig:phase-diagram}), the coupling matrix is negative semi-definite, and the partition function only has a pole at $\beta = -\frac{18}{\mu_2}$. At low temperatures, the ground state dominates, and the system is in the FML phase. As we increase the temperature, once $\beta$ passes through $-\frac{18}{\mu_2}$, the system goes into a high-temperature paramagnetic phase. When we keep decreasing $\beta$ passing $0$ and onward, we do not encounter other poles in the partition function, therefore we expect the low-temperature phase on the other side of $\beta$, or equivalently, flipping the signs of all couplings, should still be a paramagnetic phase. This is consistent with the phase at $\mu_1=0$ and $\mu_2=1$ (shallow blue dot in \cref{fig:phase-diagram}). On the other hand, in the case of the neutrino model, we have $\mu_1>0$ and $\mu_2<0$ (dark red dot in \cref{fig:phase-diagram}). The coupling matrix has both positive and negative eigenvalues, and the partition function has poles at both positive and negative $\beta$. Therefore we expect another phase transition at negative $\beta$. Indeed, we see that when $\mu_1<0$ and $\mu_2>0$ (shallow red dot in \cref{fig:phase-diagram}), the system is in a ferromagnetic phase, which is different from the high-temperature paramagnetic phase.

We can illustrate numerically the behavior of the order parameter (Eq. \eqref{eq:classical_order_param}) as a function of temperature in Fig. \ref{fig:order_cl_neutrino} for the classical neutrino Hamiltonian, in the case $\mu_1=-\mu_2=2$. In the high-temperature regime, it is approximately zero, then
rapidly climbs after the phase transition. As the temperature drops further it slowly asymptotes to one as the system cools to the ground state. 

\begin{figure}
  \center\includegraphics[scale=0.5]{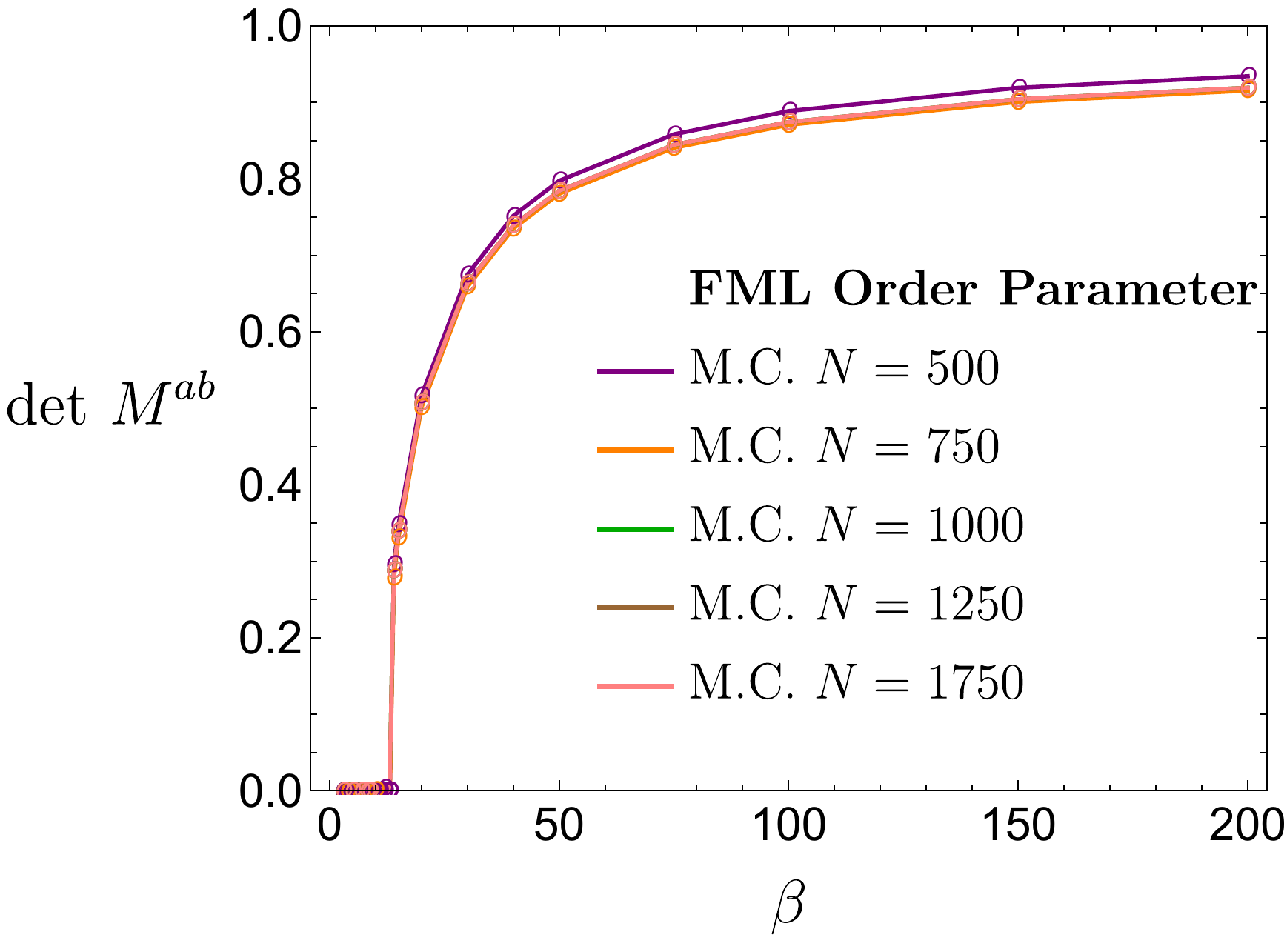}
  \caption{\label{fig:order_cl_neutrino} The FML order parameter as a function of the inverse temperature for the classical neutrino model. We see that once we pass the predicted phase transition, the parameter rapidly rises and then asymptotically approaches 1 as the temperature approaches zero.}
\end{figure}

\subsection{Quantum analysis}
We now wish to understand the phase below temperature $\beta=-2n_f d/\mu$ for the partition function in Eq. \eqref{eq:ex_part}.  To do so, we focus on the case that $d=3$, and we form the product of coherent states: 
\begin{align}\label{eq:flavor_momentum_locked}
    |\{\vec{S}_k\}_{k=1}^{N}\rangle = \bigotimes_{k=1}^{N} \e^{\i \pi \vec{S}'_k \cdot \frac{\vec{\sigma}_k}{2}} |0\rangle_k,
\end{align}
where $\vec{S}'$ is related to $\vec{S}$ by reducing its polar angle by half.

Then we have, making use of Eq. \eqref{eq:pauli_complete} and setting $\vec{S}_k=\vec{v}_k$:
 \begin{align}
\langle\{\vec{v}_k\}_{k=1}^{N}|\hat{\vec{\sigma}}_{i}|\{\vec{v}_k\}_{k=1}^{N}\rangle&=\vec{v}_i\,,\\
\langle\{\vec{v}_k\}_{k=1}^{N}|\hat{H}^{\rm rot.}_{\rm q}|\{\vec{v}_k\}_{k=1}^{N}\rangle&=\langle\{\vec{v}_k\}_{k=1}^{N}|\hat{H}^{\nu\nu}_{\rm q}|\{\vec{v}_k\}_{k=1}^{N}\rangle= \sum_{1\leq i < j \leq N} \frac{\mu}{2 N}\big(\vec{v}_{i}\cdot\vec{v}_{j}\big)^2\approx -\frac{\mu N}{4\times 3}\,,
\end{align}  
Our order parameter for the phase defined by this state is constructed from the tensor:
\begin{align}
\hat{M}^{ab}=\frac{d}{N}\sum_{i=1}^{N}v_i^{a}\hat{\sigma}^{b}_{i}\,,
\end{align}
Taking the expectation value in the state \eqref{eq:flavor_momentum_locked} we get:
\begin{align}
\langle\{\vec{v}_k\}_{k=1}^{N}|\hat{M}^{ab}|\{\vec{v}_k\}_{k=1}^{N}\rangle=\delta^{ab}.
\end{align}
If we perform a global $\SU(2)$ rotation on the state, the order parameter changes as:
\begin{align}
\langle\{\vec{v}_k\}_{k=1}^{N}|\hat{U}^{\dagger}\hat{M}^{ab}\hat{U}|\{\vec{v}_k\}_{k=1}^{N}\rangle=O^{ab},
\end{align}
where $O \in \SO(3)$.
We can simply mod out such rotations by considering the determinant on the indices $a,b$:
\begin{align}\label{eq:quantum_order_param}
    \hat{M}=\text{det}\hat{M}^{ab}\,.
\end{align}
This is the generalization of Eq.~\eqref{eq:classical_order_param}.

In disordered states away from the vectors $\vec{v}_i$, the expectation value of $\hat{M}^{ab}$ vanishes.
\subsubsection{Solution of Classical Eqn. of Motion}
Here we do not write down the coherent state path integral for the neutrino system, but we refer the reader to Ref. \cite{Balantekin:2006tg}. For our purposes, we simply need the saddle points to be given by the solutions to the classical equations of motion, derived with the correct Poisson bracket to respect the $\SU(2)$ symmetry of the system. If these solutions are stable, then the saddle points are minima, and we have found a configuration that dominates the ground state.

Focusing on the neutrino case of $\mu_1=-\mu_2=\mu$, the classical equations of motion are given as: 
\begin{align}\label{eq:full_classical_eq_motion}
    \frac{\partial}{\partial t}\vec{S}_{i}&=\sum_{j=1,\neq i}^N\frac{\mu}{2N}(1-\vec{v}_i\cdot\vec{v}_j)\vec{S}_{i}\times\vec{S}_{j}\,.
\end{align}
This is seen either from the classical limit of the quantum system, or demanding that the rotors $\vec{S}_i$ satisfy the Poisson bracket:
\begin{align}
    \{S_i^{a},S_j^b\} = \delta_{ij}\sum_{c=1}^{3}\varepsilon^{abc}S^c_i\,,
\end{align}
with the Levi-Civita symbol $\varepsilon^{abc}$ providing the structure constants of the Lie Group $\SO(3)$. If we expand $\vec{S}$ near the FML state as $\vec{S}_i=r_i\vec{v}_i+\vec{\theta}_i$, and substitute it into the equations of motion with $\vec{\theta}_i$ small and solving $r_i$ to constrain $\vec{S}_i$ to be on the sphere, we get:
\begin{align}
    \frac{\partial}{\partial t}\vec{\theta}_{i} &= \sum_{j=1,\neq i}^N\frac{\mu}{2N}(1-\vec{v}_i\cdot\vec{v}_j)\big(\vec{v}_{i}\times\vec{\theta}_{j}-\vec{v}_{j}\times\vec{\theta}_{i}\big)+O(\theta^2)
\end{align}
The dangerous term is $\vec{v}_{i}\times\vec{\theta}_{j}$, as this can generate forces that would move $\vec{\theta}_i$ in directions away from the $\vec{v}_i$ neighborhood. To see how we need not worry, we write out the explicit indices on the above Eq.:
\begin{align}
    \frac{\partial}{\partial t}\theta^a_{i}&=\sum_{j=1,\neq i}^N\frac{\mu}{2N}(1-\sum_{e=1}^{3}v^e_i v^e_j)\Big(\sum_{b,c=1}^{3}\varepsilon^{abc}v^b_{i}\theta^c_{j}-\varepsilon^{abc}v^b_{j} \theta^c_{i}\Big)+O(\theta^2)
\end{align}

But we note that $\vec{v}_j$ is assumed perpendicular to $\vec{\theta}_j$, and at any given time, we can assume it is essentially in a random orientation about $\vec{v}_j$. So when we perform the sum over $j$, we conclude this term will vanish:
\begin{align}
\sum_{j}v_j^{e}\theta_j^c\rightarrow N(0+O(N^{-1/2}))\,.
\end{align}
We conclude by then applying the rule in Eq.\eqref{eq:momentum_ave_rules}:
\begin{align}
\frac{\partial}{\partial t}\vec{\theta}_{i}&\approx \frac{\mu}{6}\vec{v}_{i}\times\vec{\theta}_{i}+O(\theta^2)+O(N^{-1/2})
\end{align}
Thus $\vec{\theta}_i$ will precess about $\vec{v}_i$, and the potential forces which could drive it away from a neighborhood about $\vec{v}_i$ average to zero at infinite $N$. At finite $N$, they do not exactly vanish, but cannot act coherently to push $\vec{\theta}_i$ far out of the initial neighborhood of $\vec{v}_i$, but do allow it to explore the whole neighborhood. We can illustrate this fact by numerically solving\footnote{Using the standard Runge-Kutta-Fehlberg (RKF45) method, with a local error tolerance of $5\times 10^{-6}$. } Eq. \eqref{eq:full_classical_eq_motion} for $N=750$ neutrinos with $\mu=1$. We start in each neutrino in a flavor state that is within an angle of $10^{-3}$ of its corresponding momentum, and evolve till $t_f=500$, yielding an order parameter that is $O=0.95$. We show the conformal (stereographic) projection of the first neutrino's state trajectory onto the plane tangent to the sphere at the direction $\vec{v}_i$ for $i=1$ and $750$ in Fig. \ref{fig:classical_fml_evo}. The $\theta_x$ and $\theta_y$ axes are arbitrary directions that parametrize the tangent plane to the sphere at the point defined by each neutrino's momentum, and are not defined with any absolute orientation.

In the left panel of \cref{fig:classical_fml_evo_compare}, we give the comparison of the FML trajectory with order parameter $O=0.95$ to that of trajectory with an order parameter $O=0.45$, using the same couplings now only for the first neutrino. We can see that the initial conditions with $O=0.45$ do not stay confined to a limited neighborhood about the momentum direction, and very rapidly change in time. In the right panel of \cref{fig:classical_fml_evo_compare}, we give the comparison of the FML trajectory with order parameter $O=0.95$ to that of a random initial state, using the same couplings and again for the first neutrino, but now evolved to $t_f=2000$. The random initial state has an order parameter that is essentially zero. The random initial conditions now evolve much more slowly than the trajectory with $O=0.45$, but are not confined as the order parameter $O=0.95$ trajectories. 

\begin{figure}
  \center\includegraphics[scale=0.33]{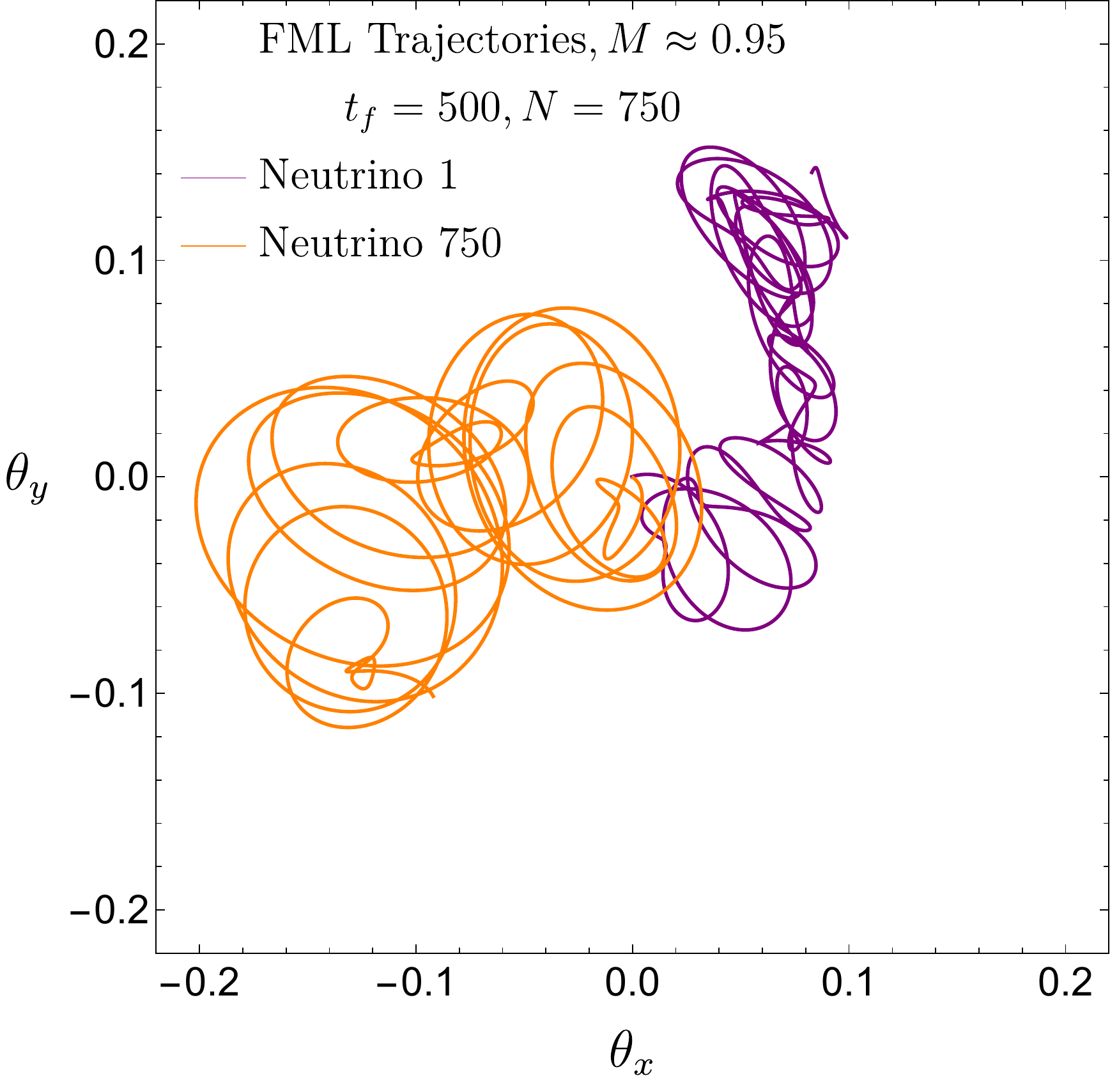}
  \caption{\label{fig:classical_fml_evo} Classical trajectories of the neutrino 1 and 750 when the initial state is close to saturating the order parameter of Eq. \eqref{eq:classical_order_param}. Note that the trajectories stay confined to a small neighborhood about the corresponding momentum direction. We plot the stereographic projection of the flavor-sphere of the neutrino onto the plane tangent to the point given by the momentum direction.  The orientation of the tangent planes is random for the purposes of plotting.}
\end{figure}

\begin{figure}
 \centering
    \begin{subfigure}[t]{0.4\textwidth}
        \centering
        \includegraphics[width=0.81\textwidth]{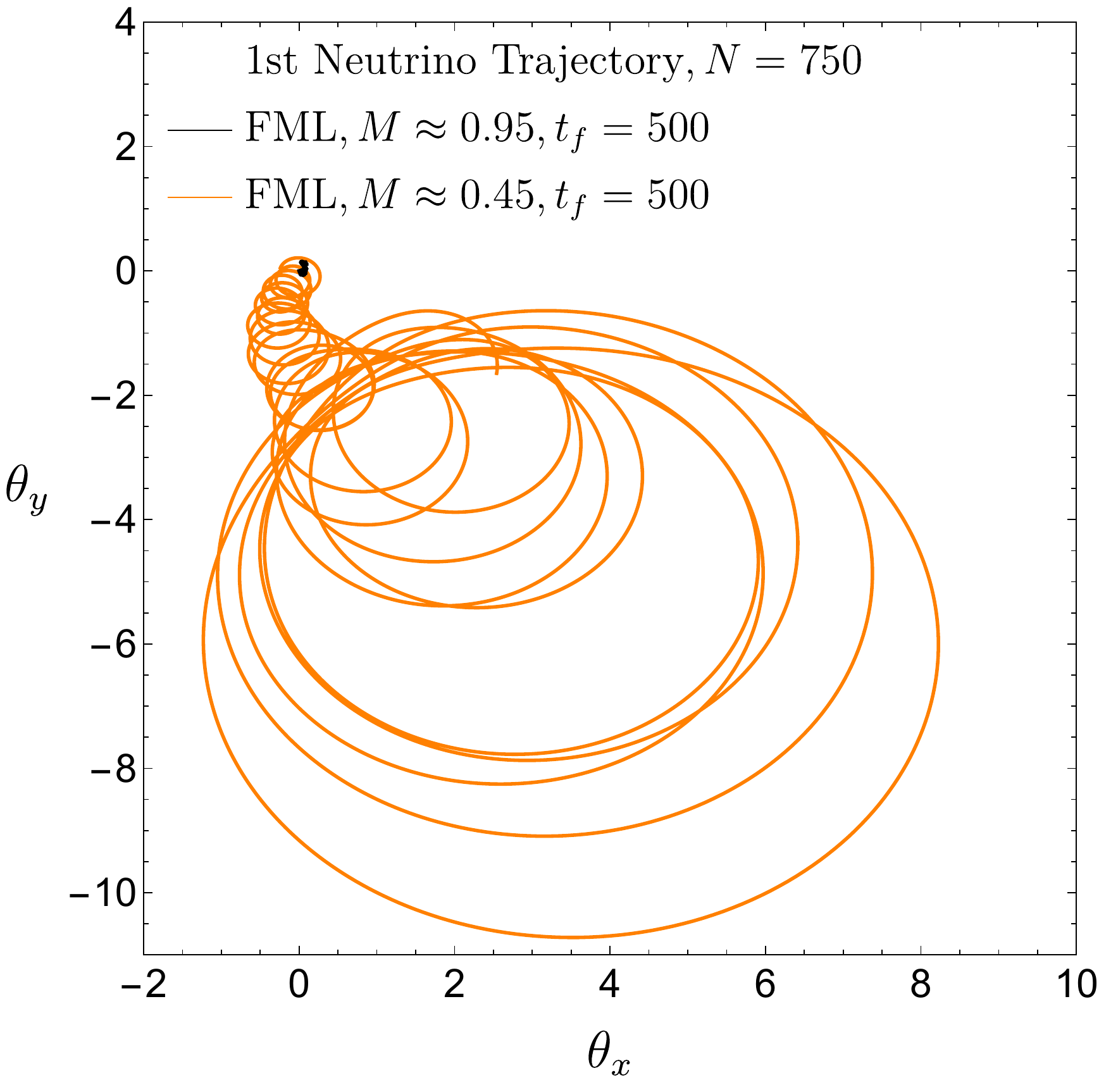}
        \label{fig:classical_fml_evo_compare_almost}
    \end{subfigure} \quad
    \begin{subfigure}[t]{0.4\textwidth}
        \centering
        \includegraphics[width=0.78\textwidth]{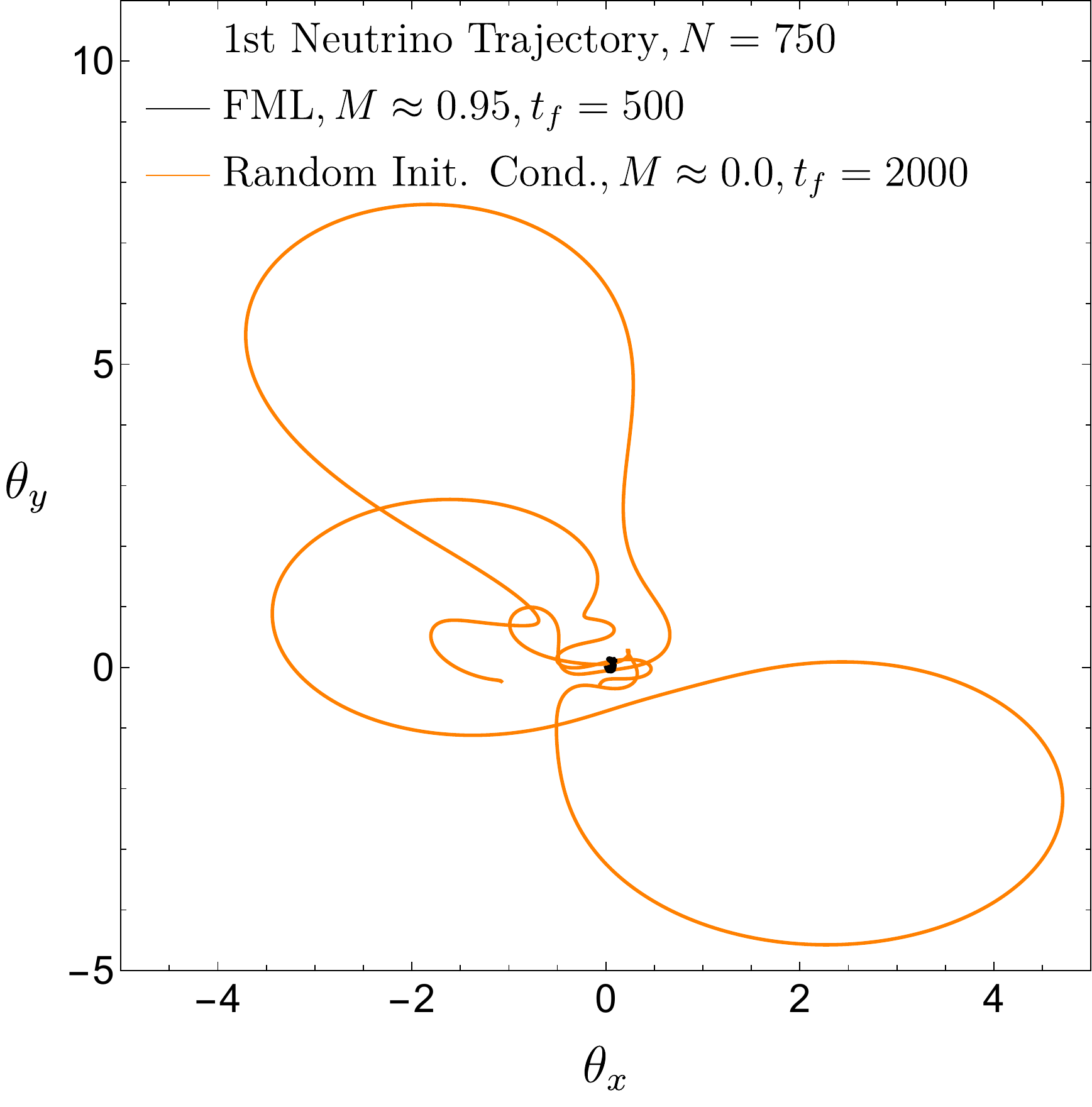}
        \label{fig:classical_fml_evo_compare_non_fml}  
    \end{subfigure} \caption{Comparison of the classical trajectories of the neutrino 1 for two distinct initial states: The black curves in both panels are when the initial state is close to saturating the order parameter of Eq. \eqref{eq:classical_order_param}. We plot the stereographic projection of the flavor-sphere of the neutrino onto the plane tangent to the point given by the momentum direction. Left panel: the orange curve is when the order parameter is $\approx 0.45$, and thus there is a fair degree of correlation between the momenta and flavor, but not enough to confine the trajectory. Right panel: the orange curve is when the order parameter is $\approx 0$, and with no correlation between the momenta and flavor, the trajectory evolves slowly.}    
    \label{fig:classical_fml_evo_compare}
\end{figure}

\section{Conclusions}
Motivated by the forward scattering model of neutrino interactions in supernovae and the fact that the late-time expectation values of observables should be determined by thermal ensembles, the so-called eigenstate thermalization hypothesis, we have given a sketch of the phase diagram for ``asymptotically'' simple Heisenberg magnets with all-to-all connectivity. By asymptotically simple, we mean that the eigenvalues of the coupling matrix are never too large or too many in comparison to the large-$N$ limit. We calculated the quantum partition function in the high-temperature limit for a model possessing a large degree of symmetry (in both the coupling distribution and the Hamiltonian), yet is still quantum chaotic, see Ref.~\cite{Martin:2023gbo}, and established a phase transition to the so-called flavor-momentum-locked (FML) state, where there is a large correlation between the directions used to define the random coupling matrix and states of the magnet/neutrinos.
 
For these simple magnets, we can show the largest eigenvalue induces a pole in the classical partition function, implying that even in the quantum case, we must have a phase transition via the Leib inequalities. We gave some example phases for this transition, including the FML phase, and studied the zero-temperature phase diagram. Going forward, one would like to know whether this FML phase generalizes: does the matrix defined on the sites $\vec{S}_i\cdot\vec{S}_j$ lock onto some suitably defined projection of the coupling matrix $J_{ij}$ as we lower the temperature, and how does the number of eigenvalues of the coupling matrix control this phase-transition?

Perhaps our most surprising result, from a condensed-matter point of view, is the construction of classical models whose density of states is manifestly \emph{not} Gaussian up to the phase transition. Further, we have no reason to expect the corresponding quantum density of states to be Gaussian, as this would be difficult to square with the Leib inequalities on the partition function. We note we do not violate established theorems of Refs. \cite{MON197590} and \cite{2014MPAG...17..441E}, since there the coupling matrices manifestly fail our condition on the number of eigenvalues. This does suggest an interesting new question: do we have a type of phase transition or cross-over in the shape of the density of states, like that outlined in Ref.~\cite{2014MPAG...17..441E}, as we tune the size or number of eigenvalues of the coupling matrix, and what is the precise condition on the coupling matrix for this to occur? 

With regards to neutrino physics, an important future consideration is to work out the relationship between the FML phase and the ``fast-flavor'' instabilities, discussed in Refs.~\cite{Abbar:2017pkh,Yi:2019hrp,Morinaga:2021vmc,Fiorillo:2023mze}. Due to the degree of coherence between the initial flavor state and the momentum directions, when we are close to the phase transition as measured by the order parameter, we have found that each neutrino can rapidly explore its available flavor space classically, while deep in the FML phase, the flavor is locked to a neighborhood defined by the momentum, though evolution is still ``fast''. No assumptions of axial symmetry, matter backgrounds, or integrability are needed for the rapid flavor evolution, and indeed, the underlying distribution of couplings we chose leads to chaos in general, Ref.~\cite{Martin:2023gbo}. 

\section{Acknowledgements}
We'd like to thank Scott Lawrence, Yukari Yamauchi, and Joe Carlson for many conversations, and Francesco Caravelli for discussions on spin glasses. This work was supported by the Quantum Science Center (QSC), a National Quantum Information Science Research Center of the U.S. Department of Energy (DOE) and by the U.S. Department of Energy, Office of Science, Office of Nuclear Physics (NP) contract DE-AC52-06NA25396.

\bibliography{biblio}

\appendix

\section{Conjugacy classes of permutations and their generating functions}\label{sec:permutation_review}
In this appendix, we review some basic facts of the conjugacy classes of permutations and derive the generating function \cref{eq:partition_classical_rotor}.

\subsection{Basic facts of the conjugacy classes of permutations}\label{sec:permutation_facts}

Fix $L$ objects to be permuted, and we denote the set of all such permutations as $\Sym{L}$. Any such permutation can be decomposed into disjoint cycles. The cycle structure is preserved under conjugation, and can be used to label the conjugacy classes of the permutations. That is, let $\sigma$ and $\rho$ be permutations, so that:
\begin{align}
    \rho=(a_{11}\cdots a_{1\lambda_1})(a_{21}\cdots a_{2\lambda_2})\cdots (a_{k1}\cdots a_{k\lambda_k}),\\
    \sigma^{-1}\rho\sigma=\big(\sigma(a_{11})\cdots \sigma(a_{1\lambda_1})\big)\cdots \big(\sigma(a_{k1})\cdots \sigma(a_{k\lambda_k})\big)\,.
\end{align}
Each $a_{ij}$ is in the set of $L$ objects and $\rho$ maps $a_{ij}$ to $a_{i (j\mathrm{mod}\lambda_i+1)}$. Thus the sequence of integers $[\lambda_{1},\cdots,\lambda_{k}]$, with $\lambda_{1}\leq\cdots\leq\lambda_{k}$, specifies a conjugacy class. When $\lambda_{i} \geq 2$ for all $i$, the permutation is called a derangement, which has no fixed point. We use $\derange{L}$ to denote the set of the derangement conjugacy classes. Thus, for example, we have: $\derange{6}=\{[2,2,2],[2,4],[3,3],[6]\}$.

We note that any conjugacy class can also be specified by counting its number of disjoint cycle types. When a permutation is in the symmetric group on $L$ elements, then we can also use the notation $[1^{\alpha_1}2^{\alpha_2}\cdots L^{\alpha_L}]$ to specify the conjugacy class. This notation means that permutations in that class have $\alpha_\ell$ number of disjoint cycles of length $\ell$. This implies the constraint $L=\sum_{\ell=1}^{L}\ell\alpha_{\ell}$. Thus we have in our $L=6$ example, we have $[2,2,2]\equiv [1^{0}2^{3}3^{0}4^{0}5^{0}6^{0}]$. In general, we have two ways of specifying the conjugacy class the relations
\begin{align}
  [\lambda_1,\cdots,\lambda_{k}]&\equiv [1^{\alpha_1}2^{\alpha_2}\cdots L^{\alpha_L}]\,,\\
  \sum_{\ell=1}^{L}\alpha_{\ell}&=k\,,\\
  \sum_{\ell=1}^{L}\ell\alpha_{\ell}&=\sum_{i=1}^{k}\lambda_i=L\,,\\
  \Big|[\lambda_1,\cdots,\lambda_k]\Big| &= \frac{\Gamma(1+L)}{\Gamma(1+k) \prod_{i=1}^{k} \lambda_i} , \label{eq:k-cycle-size}\\
  \Big|[1^{\alpha_1}2^{\alpha_2}\cdots L^{\alpha_L}]\Big| &= \frac{\Gamma(1+L)}{\prod_{\ell=1}^{L}\ell^{\alpha_\ell}\Gamma(1+\alpha_\ell)}\,.
\end{align}
In the last line, we give the number of elements in the specified conjugacy class. The number of derangements of order $L$ that contains exactly $k$ cycles is known as associate Stirling numbers of the first kind, denoted as $d(L,k)$. 

Using the generating function of $d(L,k)$ weighted by the coupling matrix is a crucial step in summing the high-temperature series expansion in \cref{eq:partition_classical_rotor}. In order to illustrate the idea with simpler calculations, we review the derivation of the generating function of $d(L,k)$ in the next subsection. The generating function weighted by the coupling matrix then directly follows.

\subsection{The generating function of $d(L,k)$}\label{sec:gen_func_stirling}

In this subsection, we derive the generating function of $d(L,k)$, i.e., associate Stirling numbers of the first kind, using the composition of generating functions. Since any permutation can be decomposed into a product of cycles, the basic building blocks are single cycles, whose number of elements is counted by $d(L,1)$. Therefore we derive the generating function of $d(L,1)$ first. Since the rotation of a cycle does not change the permutation, the number of $L$-cycles is $d(L,1) = (L-1)!$. However, since we will not count 1-cycles, we have to manually set $d(1,1)=0$. Therefore the generating function of $d(L,1)$ is 
\begin{align}
    \sum_{L=2}^{\infty} \frac{d(L,1)}{L!} z^L = \sum_{L=2}^{\infty} \frac{1}{L} z^L - z = -\ln(1-z) - z.
\end{align}

$|[\lambda_1,\cdots,\lambda_k]|$ counts the number of elements with cycle structure $[\lambda_1,\cdots,\lambda_k]$. Therefore we can write
\begin{align}
    \big|[\lambda_1,\cdots,\lambda_k]\big| = \frac{1}{k!} \binom{L}{\lambda_1, \cdots \lambda_k} d(\lambda_1,1) \cdots d(\lambda_k,1),
\end{align}
where $\binom{L}{\lambda_1, \cdots \lambda_k}$ is the multinomial coefficient that counts the ways to put $L$ numbers into $k$ cycles with size $\lambda_1, \cdots \lambda_k$. The factor $\frac{1}{k!}$ is because the order of the cycles does not matter. It can be easily checked that this expression is identical to \cref{eq:k-cycle-size}.

$d(L,k)$ counts the number of derangements with $k$ cycles whose total length is $L$. Therefore $d(L,k)$ is simply a sum over $[\lambda_1,\cdots,\lambda_k]$ with $\lambda_i \geq 2$, 
\begin{align}
    d(L,k) = \sum_{\substack{\lambda_1+\cdots+\lambda_{k}=L\\\lambda_k\geq\cdots\geq \lambda_1\geq 2}} \big|[\lambda_1,\cdots,\lambda_k]\big| = \frac{1}{k!} \sum_{\substack{\lambda_1+\cdots+\lambda_{k}=L\\\lambda_k\geq\cdots\geq \lambda_1\geq 2}} \binom{L}{\lambda_1, \cdots \lambda_k} d(\lambda_1,1) \cdots d(\lambda_k,1),
\end{align}

Therefore, it can be checked that the generating function of $d(L,k)$ is $k$-th power of the generating function of $d(L,1)$, divided by $k!$,
\begin{align}
    \sum_{L=2}^{\infty} \frac{d(L,k)}{L!} z^L = \frac{1}{k!} \Big(\sum_{L=2}^{\infty} \frac{d(L,1)}{L!} z^L \Big)^k = \frac{1}{k!} \big(-\ln(1-z) - z\big)^k.
\end{align}
Summing over all $k$, we have
\begin{align}\label{eq:gen_func_stirling}
    1 + \sum_{k=1}^{\infty} \sum_{L=2}^{\infty} \frac{d(L,k)}{L!} z^L u^k = \sum_{k=0}^{\infty} \frac{1}{k!} \big(-\ln(1-z) - z\big)^k u^k = \e^{-u\ln(1-z) - uz} = (1-z)^{-u} \e^{-uz}.
\end{align}

\subsection{The generating function weighted by the coupling matrix}

When deriving the high-temperature partition function \cref{eq:partition_classical_rotor} for the dynamical rotors, we sum over the series expansion with each cycle structure weighted by $\prod_{i=1}^k \tr(J^{\lambda_i})$. Deriving such a generating function only requires minor modifications to the calculation in the previous subsection. Again, let's start with the case of single cycles,
\begin{align}
    \sum_{L=2}^{\infty} \frac{d(L,1) \tr(J^L)}{L!} z^L = \tr\Big(\sum_{L=2}^{\infty} \frac{d(L,1)}{L!} (zJ)^L \Big) = \tr(-\ln(1-zJ) - zJ).
\end{align}
In the case of $k$ cycles, we have
\begin{align}
    \sum_{L=2}^\infty \frac{1}{L!} \sum_{\substack{\lambda_1+\cdots+\lambda_{k}=L\\\lambda_k\geq\cdots\geq \lambda_1\geq 2}} \big|[\lambda_1, \cdots, \lambda_k] \big| \prod_{i=1}^k \tr(J^{\lambda_i}) z^L &= \sum_{L=2}^\infty \frac{1}{L!} \sum_{\substack{\lambda_1+\cdots+\lambda_{k}=L\\\lambda_k\geq\cdots\geq \lambda_1\geq 2}} \frac{1}{k!} \binom{L}{\lambda_1, \cdots \lambda_k} \prod_{i=1}^k d(\lambda_i,1) \tr(J^{\lambda_i}) z^L \nonumber\\
    &= \frac{1}{k!} \Big(\sum_{L=2}^{\infty} \frac{d(L,1) \tr(J^L)}{L!} z^L \Big)^k.
\end{align}
Summing over $k$, we have
\begin{align}\label{eq:matrix-gen-fun}
    1 + \sum_{k=1}^{\infty} \sum_{L=2}^\infty \frac{1}{L!} & \sum_{\substack{\lambda_1+\cdots+\lambda_{k}=L\\\lambda_k\geq\cdots\geq \lambda_1\geq 2}} \big|[\lambda_1, \cdots, \lambda_k] \big| \prod_{i=1}^k \tr(J^{\lambda_i}) z^L u^k = \exp\Big( u\tr\big(-\ln(1-zJ) - zJ\big) \Big) \nonumber\\
    &= \exp\Big( \sum_{\ell=1}^M u \big(-\ln(1-z e_\ell) - z e_\ell \big) \Big) = \prod_{\ell=1}^M (1-z e_\ell)^{-u} \e^{-uz e_\ell} ,
\end{align}
where $e_\ell: \ell = 1, \cdots, M$ are the eigenvalues of $J$. Note that this summation is convergent only when $|z| < (\max_{\ell}{e_\ell})^{-1}$.

\section{Proof of \cref{eq:trace_rule_for_cycles}} \label{sec:proving_trace_formula}
In this appendix, we provide a proof of \cref{eq:trace_rule_for_cycles}, which follows from the following more general result:
\begin{lemma}
    Let $\rho$ be a product of swaps $(ij)$, such that each index appears exactly twice. Then the number of disjoint cycles in $\rho$ does not depend on the order of the swaps.
\end{lemma}
\begin{proof}
    Since any order of the swap operators can be obtained by a sequence of swapping pairs of the nearest swap operators, it suffices to consider swapping one pair of swap operators. When these two swaps do not share common indices, clearly they commute. Therefore we only need to consider the case when they share a common index. More concretely, we only need to consider the products $\rho_2 (cb) (ba) \rho_1$ and $\rho_2 (ab) (bc) \rho_1$, where $\rho_{1,2}$ are other fixed product of swaps. Suppose in permutation defined by $\rho_1$, $a_1$ is mapped to $a$ and $c_1$ is mapped to $c$. If $a$ does not appear in $\rho_1$, then $a_1=a$, and similar for $c$. Further, we suppose in the permutation defined by the product of swaps in $\rho_2$, $a$ is mapped to $a_2$ and $c$ is mapped to $c_2$. Furthermore, from the assumption that every index appears exactly twice, we know that $\rho_1$ and $\rho_2$ do not contain $b$.

    In the case $\rho_2 (cb) (ba) \rho_1$, we know that $a_1$ is mapped to $c_2$, while $c_1$ is mapped to $b$, and $b$ is mapped to $a_2$, following our multiplication rules. In terms of cyclic notation, we have $\rho_2 (cb) (ba) \rho_1 = (\cdots a_1c_2 \cdots c_1 b a_2 \cdots)$. Similarly, $\rho_2 (ab) (bc) \rho_1 = (\cdots a_1 b c_2 \cdots c_1 a_2 \cdots)$. Therefore, we see that in terms of cyclic notation, commuting $(ab)$ with $(bc)$ simply moves $b$ from one place to another, which does not change the number of cycles. In particular, when $a_1c_2$ and $c_1 b a_2$ belong to the same cycle, the cycle structure is unchanged; When $a_1c_2$ and $c_1 b a_2$ belong to different cycles, one cycle length is increased by 1, while the other is decreased by 1. This concludes the proof.
\end{proof}

Now we take our set of swap operators from the large-$N$ expansion that forms a connected chord diagram of the form depicted in Fig.~\ref{fig:surviving_diagrams}:
\begin{align}
C=\{\exch{i_1 i_2}, \exch{i_2 i_3}, \,\dots\, ,\exch{i_{\lambda} i_{1}}\}\,.
\end{align}
We note that if we label the operators as:
\begin{align}
   \hat{\rho}_{1} = \exch{i_1 i_2}, \hat{\rho}_{2} =\exch{i_2 i_3},\, \dots\, , \hat{\rho}_{\lambda} = \exch{i_{\lambda} i_{1}}\,.
\end{align}
Then the sum over all permutations of the operators may be written as:
\begin{align}
    \exch{i_1 i_2}\cdot \exch{i_2 i_3}\cdot\,...\cdot \exch{i_{\lambda} i_{1}}+\text{perms}=\sum_{\sigma\in\Sym{\lambda}}\prod_{i=1}^{\lambda} \hat{\rho}_{\sigma(i)}\,.
\end{align}
We can calculate the trivial ordering simply:
\begin{align}
    (i_1 i_2)(i_1 i_2) \cdots (i_{\lambda} i_{1}) = (i_1)(i_2 i_3 \cdots i_{\lambda}).
\end{align}
Applying the lemma, we know the product of these operators in any order will have two disjoint cycles. Therefore we have
\begin{align}
    \tr \big[\exch{i_1 i_2}\cdot \exch{i_2 i_3}\cdot\,...\cdot \exch{i_{\lambda} i_{1}}+\text{perms}\big]&=\Gamma(1+\lambda)n_{f}^{N+2-\lambda},\text{ for } \lambda>2\,. \tag{re \ref{eq:trace_rule_for_cycles}}
\end{align}  

Rather than the above algebraic argument, there is a geometric argument one can give using the chord diagrams. First we consider the product $\prod_{i=1}^{\lambda-1} \hat{\rho}_{\sigma(i)}$ for some fixed permutation $\sigma$, where we drop the last swap $\hat{\rho}_{\sigma(\lambda)}$ from the product. This product of swap operators can be represented by chord diagrams of the form shown in Fig.~\ref{fig:chordcycles}. One can compute the permutation generated by the product of swaps by starting at any numbered line, and then as one moves along from left to right on the horizontal line, one changes to a new line as soon as one hits a dot, following the vertical line up or down to where it ends (swapping lines), and then continue along accordingly. Because the swap $\hat{\rho}_{\sigma(\lambda)}$ has been neglected in the product, the diagram will be planar, regardless of the ordering of swaps: there is never a return path to a line that you have previously been on. We can relabel the lines if necessary to guarantee the swap network will be planar on the chord diagram, see Proposition 2.1 of Ref.~\cite{1996q.alg.....7001B} where it is shown that a permutation of the labels of the lines does not affect the value of the weight of the chord diagram. This implies there are no sub-cycles to the permutation defined by the product of swaps given by these planar diagrams. Thus products $\prod_{i=1}^{\lambda-1} \hat{\rho}_{\sigma(i)}$ must always evaluate to a single cycle, and then multiplying by the final swap will split the cycle into two, 
\begin{align}
\label{eq:split_cycle_b}(ac_1...c_k b c_{k+1}...c_{\ell})(ab)&=(bc_1...c_k) (a c_{k+1}...c_{\ell}).
\end{align}
Going back to the braid interpretation of Ref.~\cite{1996q.alg.....7001B}, the chord diagram, regardless of how the ends are permuted and tied, our chord diagrams will give a braid that is composed of two distinct ropes.

\begin{figure}
  \centering\includegraphics[scale=0.8]{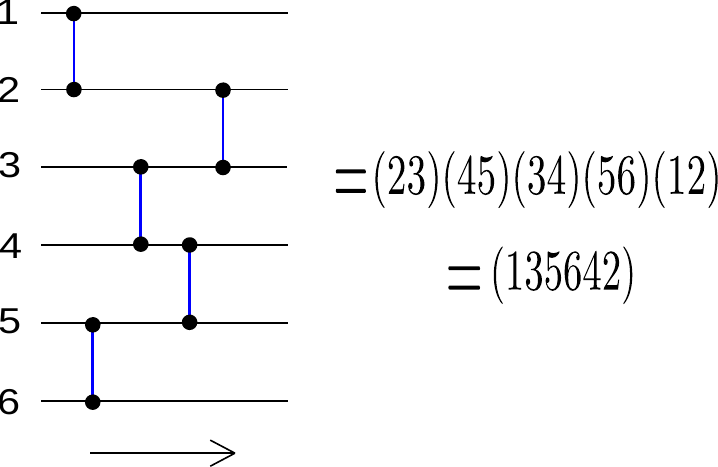}\qquad\qquad\includegraphics[scale=0.8]{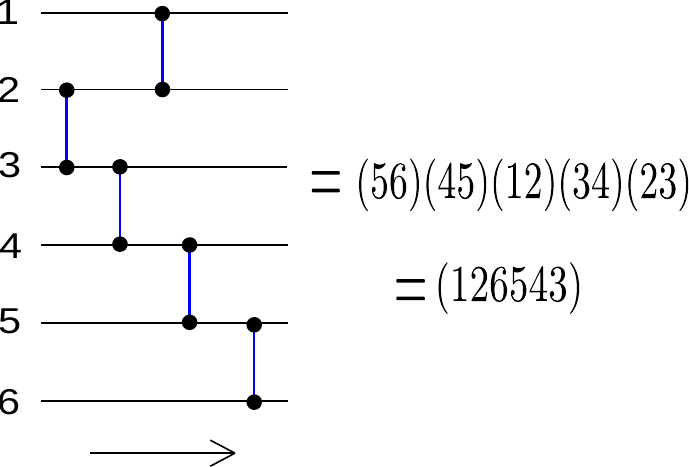}
  \caption{\label{fig:chordcycles}  We illustrate how to understand the trace rule of Eq.~\eqref{eq:trace_rule_for_cycles} in terms of chord diagrams. We delete one swap, in this case $(16)$, and we can always draw the chord diagram corresponding to the left-over product of swaps in a planar fashion. This implies we can never return to the line we started on when following the horizontal lines, here indicated by a left-to-right motion, and following the rerouting indicated by the inserted swap. Thus the product of all the other swaps must be a single cycle. Multiplying by the last swap $(16)$ will break it into two distinct cycles.}
\end{figure}

\section{Trace Moments of the Coupling Matrix and Numerical Checks}\label{sec:simple_couplings}
In this appendix, we illustrate how the formula in Eq. \eqref{eq:trace_rule_with_error} can break down 
when the number of eigenvalues scales with the number of spins. We recall the assumption we have on the coupling matrix in \cref{eq:shift_rule_for_matrix}.
We are given a list of real numbers $e_{\ell}, \ell=1,\cdots,M$, and we are given a generic orthogonal $N\times N$ matrix $O_{ij}$. With these we calculate the coupling matrix of our Hamiltonian as:
\begin{align}
J_{ij}=&\sum_{k=1}^{M}O_{ik}O_{jk}e_{k},\text{ and } \big|O_{ij}\big| \leq O_{\rm max}\sqrt{\frac{\log(N)}{N}},\,\,\forall i,j \,,\ \nonumber\\
&\text{ with }M\ll N,\text{ and }\sum_{k=1}^{N}O_{ik}O_{jk}=\delta_{ij}\,.
    \tag{re \ref{eq:shift_rule_for_matrix}}
\end{align}
and we wish to show for a permutation $\sigma$ in the conjugacy class $[\lambda_1,\cdots,\lambda_k]$ that
\begin{align}\label{eq:trace_rule_recall}
\sum_{\{i_{1},\cdots,i_{L}\}\subset\Lambda}\prod_{j=1}^{L} J_{i_ji_{\sigma(j)}}=f_{[\lambda_1,\cdots,\lambda_{k}]}(J)\approx \prod_{i=1}^{N}\text{tr}[J^{\lambda_i}]\,.
\end{align}
We then finish with some numerical checks, verifying that $M\ll N$ is indeed necessary.

\subsection{Trace Moment Inequalities of the Coupling Matrix}\label{sec:inequalities}

 We can easily show the necessity of our conditions on the coupling matrix, simply consider:
\begin{align}\label{eq:trace_rule_example}
    f_{[2]}(J)&=\tr[J^{2}]-\sum_{i=1}^{N}J_{ii}^2\,.
\end{align}
We then use Jensen's inequality:
\begin{align}
    \frac{1}{N} \sum_{i=1}^{N}J_{ii}^2\geq \frac{1}{N^2}\Big(\sum_{i=1}^{N}J_{ii}\Big)^2 = \frac{1}{N^2} (\tr J)^2 \sim \Big(\frac{M}{N}\Big)^2\,.
\end{align}
We can see the origin of our demand that $M\ll N$.
For a further illustrative example, we can look at the fourth-order sum when we set the diagonals of the coupling matrix to zero, $J_{ii}=0$:
\begin{align}\label{eq:trace_rule_analyze}
    f_{[4]}(J)&=\sum_{i\neq j\neq k\neq\ell}J_{ij}J_{jk}J_{k\ell}J_{\ell i},\nonumber\\
    &=\sum_{i,j,k,\ell}\Big(J_{ij}J_{jk}J_{k\ell}J_{\ell i}-\frac{2}{N}J_{ij}^2 J_{jk}^2+\frac{1}{N^2}J_{ij}^4\Big)\,.
\end{align}
The trace of the fourth power of the coupling matrix has the scaling:
\begin{align}
    \sum_{i,j,k,\ell} J_{ij}J_{jk}J_{k\ell}J_{\ell i} = \tr(J^4) = \sum_{\ell=1}^M e_\ell^4 = O(M).
\end{align}
We now wish to lower bound the terms in Eq. \eqref{eq:trace_rule_analyze} which prevents the identification in Eq. \eqref{eq:trace_rule_with_error}. Again to use Jensen's inequality, we define the vector $w_i=\sum_{j} J_{ij}^2$. This will show $\sum_{i}w_i^2=\sum_{i,j,k} J_{ij}^2 J_{ik}^2$ is lower bounded by $\frac{1}{N} \tr(J^2)^2 \sim \frac{M^2}{N}$. Therefore, we see that if $M \sim N$, then these two terms are actually of the same order; if $M \sim N^{1/2}$, then the first term is order $N^{1/2}$ and the second term is order $1$.

We now argue that in fact our condition is sufficient. To be explicit, we denote the matrix that actually controls the physics, the off-diagonals, as $\tilde{J}_{ij}$. We recall the assumptions on it above in Eq. \eqref{eq:shift_rule_for_matrix}.
Note that the matrix $\tilde{J}_{ij}$ with $\tilde{J}_{ii}=0$ does \emph{not} have eigenvalues $e_k$, but the matrix defined for all $i,j$ as:
\begin{align}
J_{ij}&=\sum_{k=1}^{M}O_{ik}O_{jk}e_{k},\text{ for all } i,j\,,
\end{align}
does have eigenvalues $e_k$. We then consider the simple case that the permutation $\sigma$ is a single cycle of length $\lambda$, so that we can write without loss of generality:
\begin{align}
    f_{[\lambda]}[\tilde{J}]&=\sum_{\{i_1,\cdots,i_\lambda\}}\tilde{J}_{i_1 i_2}\tilde{J}_{i_2 i_3}\cdots\tilde{J}_{i_\lambda i_1} = \sum_{\{i_2,\cdots,i_\lambda\}} \sum_{l=1, l\not\in \{i_2,\cdots,i_{\lambda}\}}^{N}\tilde{J}_{\ell i_2}\tilde{J}_{i_2 i_3}\cdots\tilde{J}_{i_\lambda \ell}\,.
\end{align}
We have frozen the sites $i_{2},\cdots,i_{\lambda}$, and will attempt to perform the sum over $i_1=\ell$ given these sites. Thus we are lead to consider the product-sum:
\begin{align}
\sum_{l=1, l\not\in \{i_2,\cdots,i_{\lambda}\}}^{N}\tilde{J}_{i_\lambda l}\tilde{J}_{l i_2}=\sum_{l=1, l\not\in \{i_2,\cdots,i_{\lambda}\}}^{N}\sum_{k_1=1}^{M}\sum_{k_2=1}^{M}O_{i_\lambda k_1}O_{\ell k_1}  O_{\ell k_2}O_{i_2 k_2}e_{k_1}e_{k_2}\,.
\end{align}
Now we have explicitly:
\begin{align}
\sum_{l=1, l\not\in \{i_2,\cdots,i_{\lambda}\}}^{N}O_{\ell k_1}  O_{\ell k_2}=\delta_{k_1 k_2}-\sum_{t=2}^{\lambda}O_{i_t k_1}  O_{i_t k_2}\,.
\end{align}
Therefore:
\begin{align}
\sum_{l=1, l\not\in \{i_2,\cdots,i_{\lambda}\}}^{N}\tilde{J}_{i_\lambda l}\tilde{J}_{l i_2}=\sum_{\ell=1}^{N}J_{i_{\lambda}\ell}J_{\ell i_{2}}-\sum_{t=2}^{\lambda}\sum_{k_1=1}^{M}\sum_{k_2=1}^{M}O_{i_\lambda k_1}O_{i_t k_1}  O_{i_t k_2}O_{i_2 k_2}e_{k_1}e_{k_2}\,.
\end{align}
This second term is bounded as:
\begin{align}
\Big|\sum_{t=2}^{\lambda}\sum_{k_1=1}^{M}\sum_{k_2=1}^{M}O_{i_\lambda k_1}O_{i_t k_1}  O_{i_t k_2}O_{i_2 k_2}e_{k_1}e_{k_2}\Big|\leq \lambda\|J\|^2O_{\rm max}^4\frac{\log(N)^2 M^2}{N^2} \,.
\end{align}
with $\|\cdot\|$ the spectral norm. Note that we have a different bound for the first term:
\begin{align}
\Big|\sum_{\ell=1}^{N}J_{i_{\lambda}\ell}J_{\ell i_{2}}\Big|\leq \|J\|^2O_{\rm max}^2\frac{M\log(N)}{N}\,.
\end{align}
Hence we can expect that our relative error in moving from the matrix $\tilde{J}$ with restricted sums to the matrix $J$ with unrestricted will have an estimate of $O(\lambda O_{\rm max}^2 \frac{\log(N)M  }{N})$. This argument can be repeated for the subsequent sums in $f_{[\lambda]}(\tilde{J})$. This also gives us a rough estimate when the high-temperature expansion should fail due to finite $N$ effects.

\subsection{Numerical Check, and Approach To Spin Glass}\label{sec:breakdown}
Finally, we examine a model that allows us to probe the limits of our high-temperature expansion in Eq.~\eqref{eq:partition_classical_rotor}, in particular Eq. \eqref{eq:trace_rule_with_error}. We define our coupling matrix as:
\begin{align}
J_{ij}&=\sum_{k=1}^{M}O_{ik}O_{jk}e_k,\qquad e_{k}=-2+4\frac{k}{M}\,.\\ 
\nonumber O_{ij}\text{ random orthogonal } & N\times N\text{ matrix, selected with Haar measure. }
\end{align}
We note that $M$ sets the rank of the coupling matrix $J_{ij}$. We consider two cases:
\begin{enumerate}
    \item $M = \text{ceil}(\sqrt{2N})\,,$
    \item $M = N\,.$     
\end{enumerate} 
The high temperature expansion from Eq. \eqref{eq:partition_classical_rotor} gives the energy per non-zero eigenvalue in the $N\rightarrow\infty$ limit:
\begin{align}\label{eq:Non_trivial_test}
    -\frac{\partial}{\partial\beta}\frac{\ln Z}{\text{rnk}(J)}&=-\frac{1}{8}\frac{\partial}{\partial\beta}\int\displaylimits_{-2}^{2}dx\Bigg(x\beta-3\ln\Big(1+\frac{x\beta}{3}\Big)\Bigg)\,\nonumber\\
    &=\frac{3}{2\beta}-\frac{9}{4\beta^2}\arctan\Big(\frac{2\beta}{3}\Big)\,.
\end{align}
Where rnk$(J)$ stands for the rank of the coupling matrix counting the number of non-zero eigenvalues. We plot the two cases given above in Fig.~\ref{fig:energy_spin_glass_and_not}, comparing against the Eq.~\eqref{eq:Non_trivial_test}. We see the case where the eigenvalues scale with the root of the system size agrees well with the result from the high-temperature expansion, while the case where the number of eigenvalues of the coupling matrix scales with system size has dramatically different behavior.

\begin{figure}
  \center\includegraphics[scale=0.7]{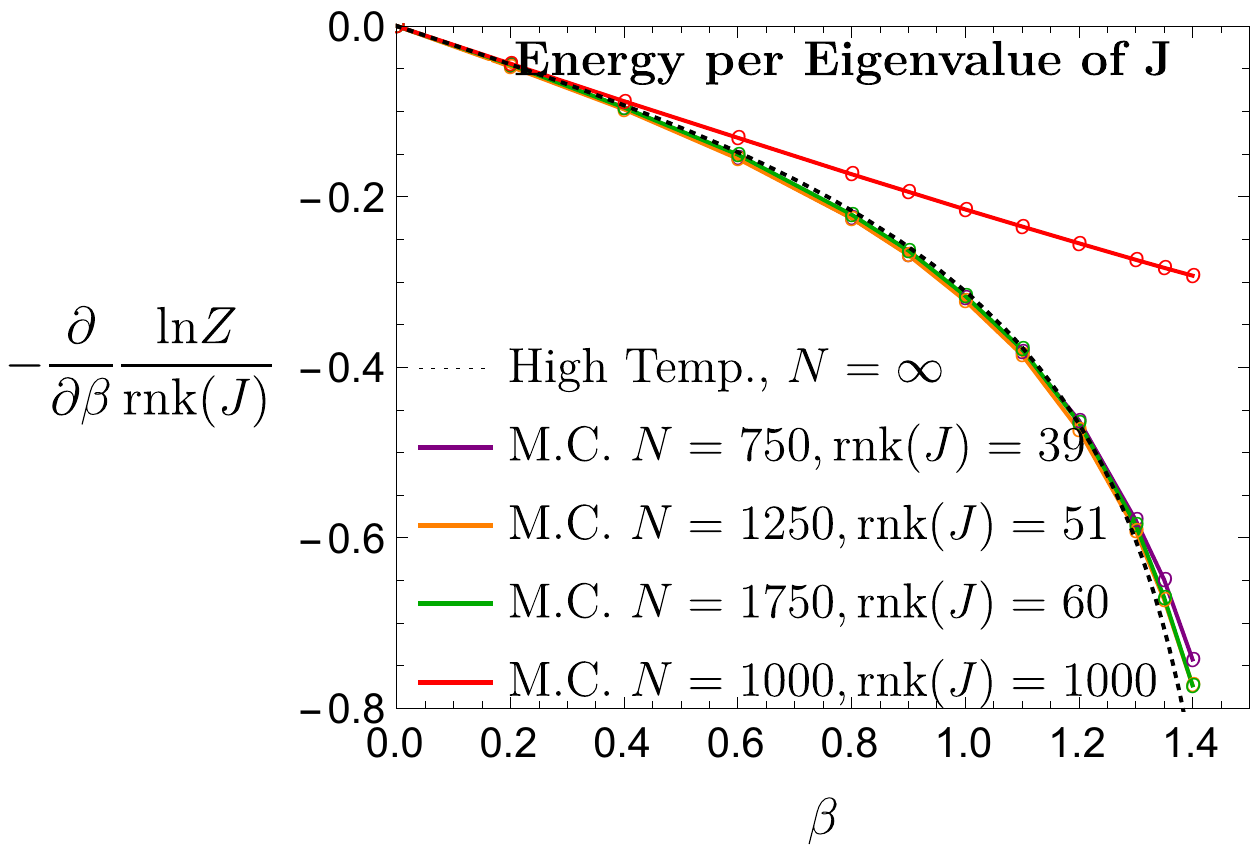}
  \caption{\label{fig:energy_spin_glass_and_not} The numerically determined energy from the heat bath algorithm. For a given $N$, the free energy is determined from a single realization of the couplings. The circles are the determination of the free energy at a specific temperature, lines are linear interpolations. The statistical uncertainty is smaller than can be shown. The black dashed line is the determination of the free energy from the high-temperature expansion, calculated in Eq. \eqref{eq:Non_trivial_test}. The red line shows the behavior when the number of eigenvalues of the coupling matrix scales as the system size, while the other curves correspond to the scale at the square root of the system size. According to the expansion given in Eq.~\eqref{eq:partition_classical_rotor}, all have the same large-$N$ limit for the energy per eigenvalue, but markedly different actual behavior.}
\end{figure}

This example shows that our formula's (Eq.~\eqref{eq:partition_classical_rotor}) validity must also depend on how the rank of the coupling matrix scales with the number of rotors. However, determining the exact conditions is beyond the scope of this paper. 
We do note that at low temperatures, the models display different behavior. While we will not give any analytic justification for this, we do show the behavior of the auto-correlation of the energy of the system as a function of the Monte-Carlo Markov Chain (MCMC) defined by the heat bath algorithm. This is defined as (Ref.~\cite{Sokal1997}):
\begin{align}
    C_{EE}(\tau)&=\frac{1}{n-\tau}\sum_{i=1}^{n-\tau}(E_i-\bar{E})(E_{i+\tau}-\bar{E})\,,\\
    \mathcal{T}_{EE}(\tau)&=\frac{1}{C_{EE}(0)}\sum_{i=1}^{\tau}C_{EE}(i)\,.
\end{align}
The index $i$ labels the successive configurations generated by the MCMC (in our case, this is a full sweep through all spins of the system, applying the heat-bath update rule to each once), that is, the Monte-Carlo time, and $n$ is the total number of configurations generated. The variable $E_i$ is the energy calculated on the $i$-th configuration, and $\bar{E}$ is the estimated average energy calculated on all generated configurations. The variable $\tau$ is the current MC time. $\mathcal{T}_{EE}(\tau)$ is the integrated autocorrelation function. One hopes this function plateaus on intermediate time scales, at a value that gives the minimum decorrelation time. Configurations that are separated by a decent multiple of this decorrelation time can be treated as statistically independent.

So in a thermalizing system, we can expect the integrated auto-correlation function to plateau on long, but not too long scales~\footnote{The integrated auto-correlation function on infinite MC time scales has infinite variance, so that fluctuations will eventually swamp its behavior, Ref.~\cite{Sokal1997}.}. If the system will not thermalize with a local-update algorithm, as seen in spin-glasses, then we can expect the integrated auto-correlation function to never stabilize to a plateau before fluctuations dominate. This is simply the statement that the algorithm is too primitive to find its way out of the local minimum on any reasonable MCMC time scale. Given this fact, the cases where the eigenvalues scale as the number of rotors appears to undergo a phase transition to such a spin-glass, as detected in Fig.~\ref{fig:correlation_time_spin_glass}. However, those where the rank of $J$ over $N$ vanishes appear not to, as the integrated auto-correlation function plateaus reasonably, Fig.~\ref{fig:correlation_time_spin_glass}.

Ref.~\cite{Baldwin:2019dki} argued that the high-temperature phase of generic spin-glasses is expected to be given by a Gaussian partition function, or rather, a Gaussian partition function cannot extend to arbitrarily low temperatures, and the low energy behavior becomes a spin-glass. We have given a recipe for designing non-Gaussian partition functions at high temperatures in Eq.~\eqref{eq:partition_classical_rotor}, but this formula appears to work precisely when there is no spin-glass phase at low temperatures. Fully fleshing this out we leave it to future investigations.

\begin{figure}
 \centering
    \begin{subfigure}[t]{0.47\textwidth}
        \centering
        \includegraphics[width=.9\textwidth]{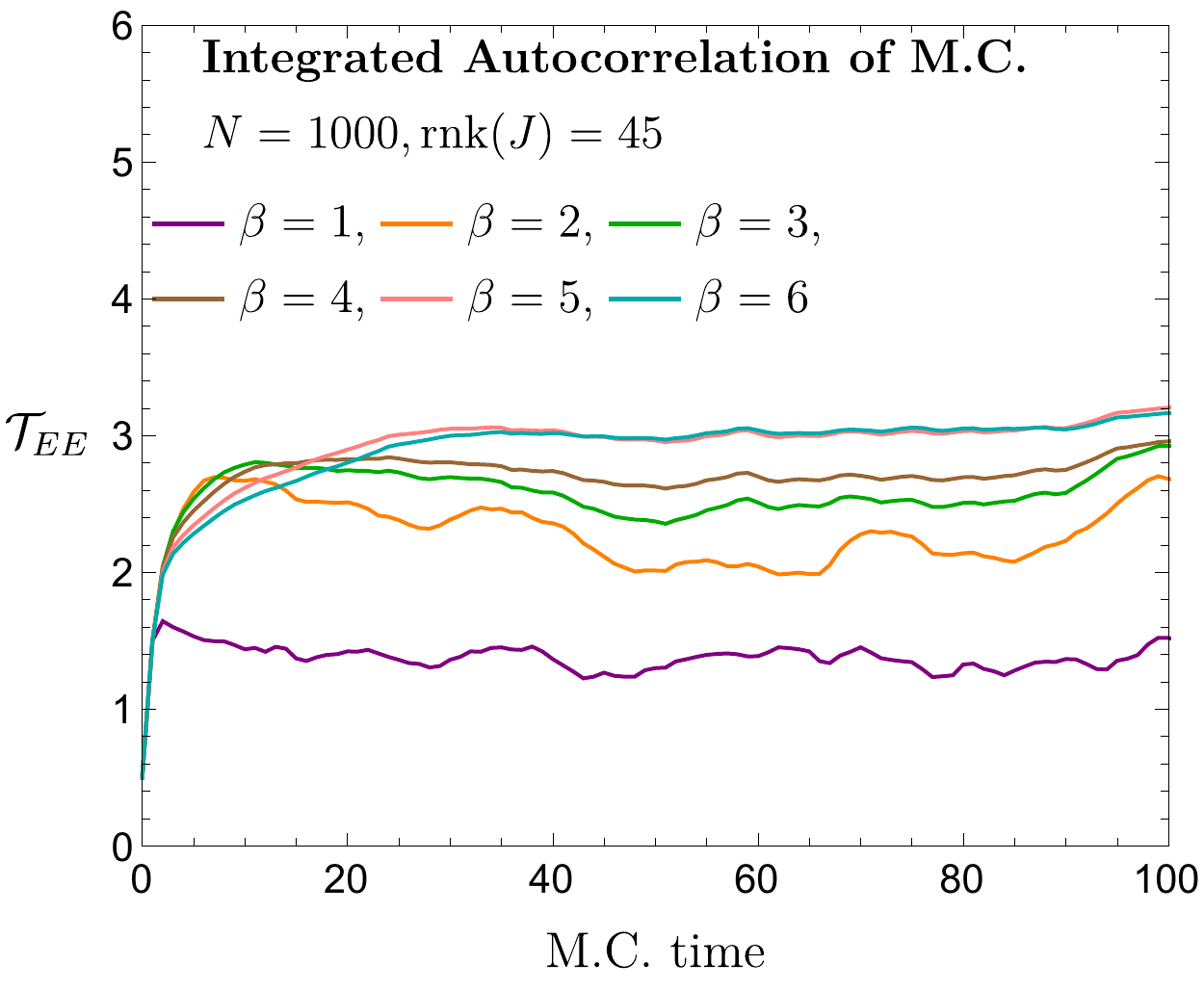}
      \caption{\label{fig:correlation_time_not_spin_glass} The integrated autocorrelation function for the MCMC is used to calculate the expected energy of the system. When a low-rank coupling matrix relative to system size is used, the integrated autocorrelation function plateaus nicely, implying that the decorrelation time between samples is well-defined.}    
    \end{subfigure} \quad
    \begin{subfigure}[t]{0.47\textwidth}
        \centering
        \includegraphics[width=.9\textwidth]{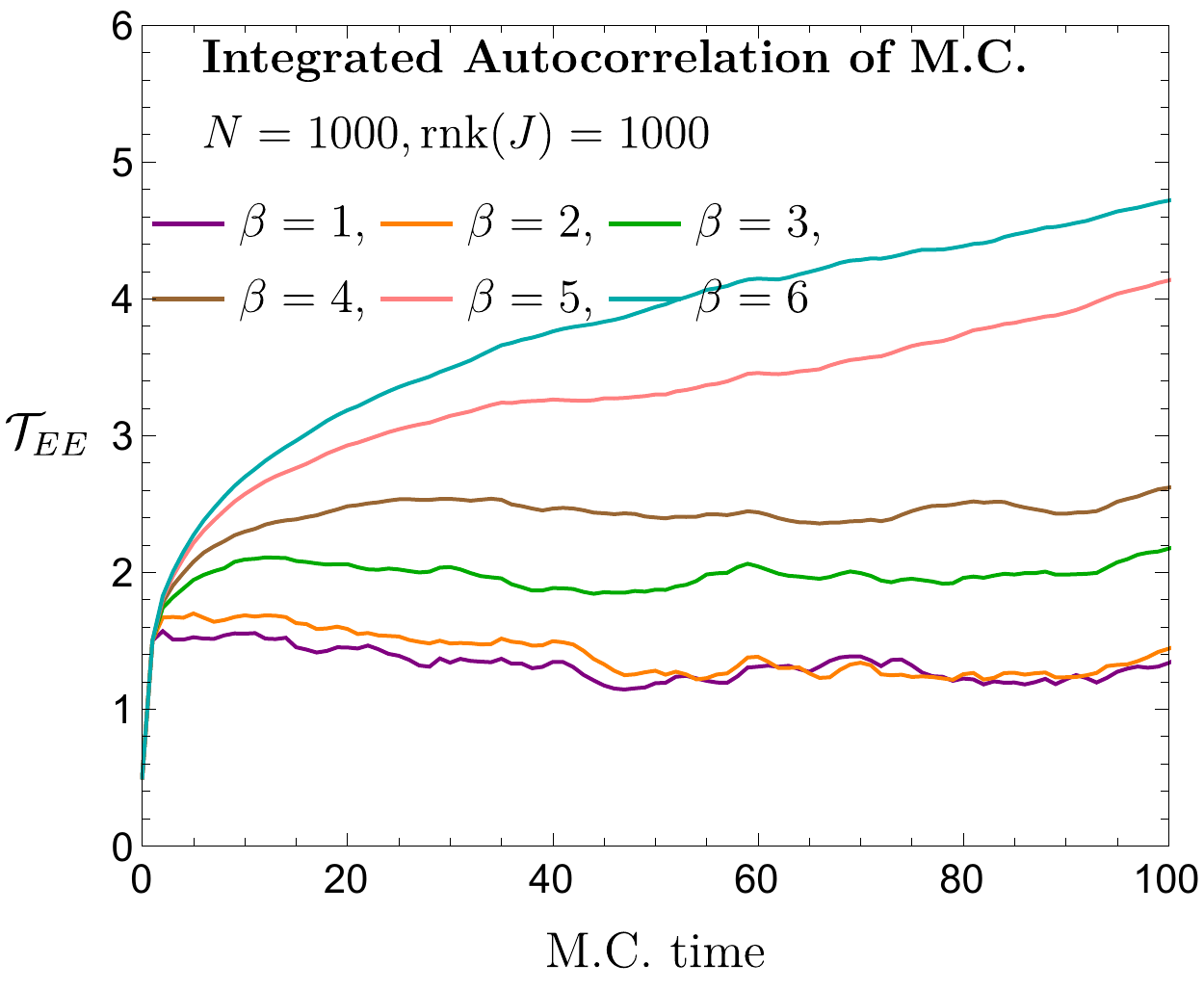}
      \caption{\label{fig:correlation_time_spin_glass}  The integrated autocorrelation function for the MCMC is used to calculate the expected energy of the system. When the rank of the coupling matrix scales with the system size, the integrated autocorrelation function, after for inverse temperatures above $\beta=4$, the integrated auto-correlation function does not plateau, suggesting a long or ill-defined decorrelation time, indicative of a spin glass. }    
    \end{subfigure} 
\end{figure}

\end{document}